\documentclass[11pt]{article}
\usepackage[a4paper,bindingoffset=0.2in,%
            left=1.2in,right=1.2in,top=1.5in,bottom=1.5in,%
            footskip=.25in]{geometry}

\usepackage[T1]{fontenc}
\usepackage[utf8]{inputenc}
\usepackage{authblk}
\usepackage[numbers]{natbib}

\usepackage{amsfonts}

\usepackage{amsthm,amsmath}
\newtheorem{lemma}{Lemma}[section]

\newtheorem{theorem}{Theorem}[section]
\newtheorem{corollary}{Corollary}[section]

\usepackage{footnote}
\usepackage{verbdef}
\usepackage{tabularx}
\usepackage{multirow, makecell}
\usepackage{xcolor}
\definecolor{ao}{rgb}{0.0, 0.5, 0.0}
\definecolor{darklavender}{rgb}{0.45, 0.31, 0.59}

\usepackage{footnote}
\usepackage{verbdef}
\usepackage{booktabs, makecell, tabularx}
\usepackage{float}

\makeatletter
\newtheorem*{rep@theorem}{\rep@title}
\newcommand{\newreptheorem}[2]{%
\newenvironment{rep#1}[1]{%
 \def\rep@title{#2 \ref{##1}}%
 \begin{rep@theorem}}%
 {\end{rep@theorem}}}
\makeatother

\newreptheorem{corollary}{Corollary}
\newreptheorem{theorem}{Theorem}
\newreptheorem{lemma}{Lemma}

\makeatletter
\newcommand\footnoteref[1]{\protected@xdef\@thefnmark{\ref{#1}}\@footnotemark}
\makeatother

\usepackage{graphicx}
\graphicspath{{images/}{../images/}}
\usepackage{booktabs}

	\newcommand{\blind}{0}
	
	\usepackage{amsmath}
	\usepackage{graphicx}
	\usepackage{enumerate}
	\usepackage{url} 
	
\begin{document}
		
		\def\spacingset#1{\renewcommand{\baselinestretch}%
			{#1}\small\normalsize} \spacingset{1}

		\if0\blind
		{
			
			\title{Unequal Opportunities in Multi-hop Referral Programs}
			
	\author[1]{Yiguang Zhang}
	\author[1]{Augustin Chaintreau}
	\affil[1]{Columbia University}

			\date{}
			\maketitle
		} \fi
		
		\if1\blind
		{

            \title{\bf \emph{IISE Transactions} \LaTeX \ Template}
			\author{Author information is purposely removed for double-blind review}
			
\bigskip
			\bigskip
			\bigskip
			\begin{center}
				{\LARGE\bf \emph{IISE Transactions} \LaTeX \ Template}
			\end{center}
			\medskip
		} \fi
		\bigskip

\begin{abstract}
As modern social networks allow for faster and broader interactions with friends and acquaintances, online referral programs that promote sales through existing users are becoming increasingly popular. Because it is all too common that online networks reproduce historical structural bias, members of disadvantaged groups often benefit less from such referral opportunities. For instance, \emph{one-hop referral} programs that distribute rewards only among pairs of friends or followers may offer less rewards and opportunities to minorities in networks where it was proved that their degrees is statistically smaller. Here, we examine the fairness of general referral programs, increasingly popular forms of marketing in which an existing referrer is encouraged to initiate the recruitment of new referred users over multiple hops. While this clearly expands opportunities for rewards, it remains unclear whether it helps addressing fairness concerns, or make them worse. We show, from studying 4 real-world networks and performing theoretical analysis on networks created with minority-majority affiliations and homophily, that the change of bias in multi-hop referral programs highly depends on the network structures and the referral strategies. Specifically, under three different \emph{constrained referral strategies} which limit the number of referrals each person can share to a fixed number, we show that even with no explicit intention to discriminate and without access to sensitive attributes such as gender and race, certain referral strategies can still amplify the structural biases further when higher hops are allowed. Moreover, when there is no constraint on the number of referrals each person can distribute and when the effect of referral strategies is removed (that is, a person's reward is linear to this person's degree in the network), we prove a precise condition under which the bias in 1-hop referral programs is amplified in higher-hop referral programs, showing that such programs are particularly concerning when the two groups have small population difference or not enough cross-group interactions. 
\end{abstract}
\maketitle
\pagestyle{plain}
\thispagestyle{empty}
\section{Introduction}
Referral programs are the marketing strategies taken by a business to motivate individual customers to share a product or service with others, and they are generally limited to (online) social networks, because effective referrals are often created by informal communications among friends\footnote{It has been shown that recommendations from friends is the most credible form of advertising among customers: https://www.nielsen.com/us/en/press-releases/2015/recommendations-from-friends-remain-most-credible-form-of-advertising/.}.
Being advocated on social networks, a user's position in a social network plays an important role on the user's access to opportunities and commercial transactions, and thus arouses concerns about the potential of reproducing bias against certain disadvantaged groups. This worry is not groundless, as existing literatures \cite{avin2015homophily} have shown that when a network contains a majority group and a minority group, the population bias against the minority group is amplified when considering the number of friends a user has. These imply that the 1-hop referral programs many online commerce companies (such as Airbnb, Uber, etc) and employee recruiting systems currently employ can reinforce the population disparity. However, much less is known about higher-hop referral programs\footnote{FantasyDraft is an example of higher-hop referral programs: https://www.referralcandy.com/blog/gaming-referral-program-examples-fantasydrafts-multi-tiered-referral-program/.}, which can result in a larger expected reach \cite{abbassi2011multi}. Having a complete characterization of social disparities in higher-hop referral programs helps tackle questions including at what point a minority group starts experiencing a systemic reproduction of disadvantage, and under which conditions -- if any -- can bias in lower hops be alleviated in higher hops.

In order to talk about these questions, we have to agree on a gaining scheme. There are numerous ways to consider the social network that serves as the foundation for the referral process. As a result, different gaining schemes are imposed by these diverse methods. For example, a referral program may only keep track of the first user who sends the referral to Alex; alternatively, the program may keep track of all users who made referrals to Alex. In this paper, we take the most straightforward and widely-used approach that considers all the referrers for each referee, as it allows us to study not only the \emph{active gain} for the referrers who benefit from getting commercial rewards, but also the \emph{passive gain} for the referees who benefit from accessing to information. The referrer-referee relationships then make up an induced graph which we refer as the \emph{referral subgraph} of the network. We define the \emph{direct (active or passive) gain} of a user $u$ in the network as the referral gain generated because of the direct neighbors of $u$, and the \emph{k-hop (active or passive) gain} (also \emph{indirect gain}) as the gains generated due to neighbors that are distance $k$ from $u$ in the referral forrest, and we define the direct gain and $k$-hop gain of a group in the network as the sum of individual direct gain and $k$-hop gain among all users belong to the group, respectively. As allowing a large number of hops has prompted increasingly illegal forms of multi-level marketing\footnote{
We acknowledge that the multi-hop referral programs also fell into the categories of multi-hop referral programs and multi-level marketing, and their legalities vary in different countries. The current paper does not take legal issues into consideration.}, such as \emph{pyramid schemes}, our work solely considers a fixed finite number of hops in the referral programs. 

We first attempt to understand the change of bias in higher-hop referral programs, among with four real-world datasets (Instagram, DBLP, Twitch, and Deezer). We empirically observe that leveraging higher-hop in referral gain strategies can amplify bias or, more interestingly, contribute to fairness. It primarily depends on which strategy is deployed and whether active or passive gains are considered. For constrained strategies where individual gains are uniformly bounded, networks with varying degrees of inherent biases and homophily agree. However, when it comes to unconstrained strategies, for which gains are effectively only bounded by the size of an individual's networks, the picture is more complex. Not only do we show that it varies with different networks, we prove mathematically sufficient and necessary conditions determining when higher hops can be effectively used to mitigate unfairness and improve referral schemes. 

The question we ask in this paper addresses the conditions at which multi-hop referral programs alleviate or reinforce the disparity in 1-hop referral programs. In particular, how can different referral strategies affect the bias, and how do social networks play a role in leading the direction of bias? To answer these questions, we introduce four referral programs, and perform mathematical analysis in a homophilous social network model that contains a majority group and a minority group. We show, under three \emph{constrained referral programs} where the number of referrals each person can share is limited, that even with no explicit intention to discriminate, 2-hop referral programs can consistently change the bias found in 1-hop, regardless of the network structure. We also show, under an \emph{unconstrained referral program} where the gains of the referral completely depend on the network structure, that the bias found in 1-hop can be amplified in higher-hop, when there is not large enough population difference between the majority group and the minority group, or when the two group have not enough interactions, and we present a precise threshold that determines the direction of bias change from 1-hop gains to k-hop gains. Finally, we extend some of our analysis for active gains to passive gains.

As a summary, we present the following contributions:
\begin{enumerate}
	\item We study real-world social networks that exhibit structural bias against specific minority groups (e.g., along gender lines or languages used) and illustrates how they relate to unequal access to referral opportunities. Most importantly, when compared with bias found in single-hop referral, we prove that bias in multi-hop referrals may be either amplified or mitigated. As it is determined by multiple factors, including a measure of homophily, our results motivates more in depth analysis (section 3).
	\item We perform mathematical analysis using homophilic network models, and we show that the bias found in 1-hop active gains can be alleviated in higher hops, even there is no intentional bias correction in referral strategies. We then provide a precise threshold condition under which 2-hop active gains reduces the bias found in 1-hop active gains. The threshold depends on the population ratio, as well as the level of the homophily in the network. We find that when the population ratio is extreme and the network is not extremely homophilic, the bias in direct active gains can be alleviated in indirect active gains (section 4).
	\item Our analysis extends to informing how to ensure fairer access to referral opportunities also for passive gains. We first prove that the case of unconstrained referral behaves similarly and presents findings for other strategies. Other extensions omitted for space reasons include considering constrained referral strategies using any number of hops $k$ greater than 2, or providing equal access to commercial opportunities among more than two groups of users. As for the result presented above, we show that extension to multi-hop are promising to mitigate unfairness but must be deployed after a careful test of the network conditions (section 5).
\end{enumerate}

This study has important ethical implications beyond referral schemes, as ensuring equitable access to information and opportunities for different demographic groups is a core ideal in any civil society. While some individual disparities, created by the network each one forms and maintains, are unavoidable, structural bias that simply results from a group size and homophilic tendencies should be minimized. As online tools stimulate information exchanges, enabling longer pathways to exchange opportunities, one is left wondering whether to celebrate that more outreach brings more equitable access, or on the contrary continue at least for some time to perpetuate long-standing differences. The answer, as our analysis proves, is that either of those viewpoints can be correct, depending on the size and homophily of various groups. Our results motivates to design network-aware referral choice that are able to deploy opportunities judiciously to leverage hop count in ways that reconcile efficiency and fairness.

\section{Related Work}
The diffusion of beliefs and conversations in social networks has been widely studied \cite{granovetter1978threshold}\cite{valente1996social}. It is also widely-considered in various virtual marketing problems, such as influence maximization \cite{domingos2001mining}\cite{kempe2003maximizing}, or revenue maximization \cite{hartline2008optimal}\cite{akhlaghpour2010optimal}\cite{abbassi2011multi}\cite{emek2011mechanisms}. These problems primarily focus on the mechanisms that maximize the benefit of the sale companies, while much less is known for how such strategies would impact the disparities in networks with more than one demographic groups.

In the meantime, the fast development of online services has risen the concern of exacerbating prejudices against historically under-represented groups, and a growing body of literatures has shown that algorithms run on big data can reproduce the prejudices even without explicit intention. Existing literatures that show evidence of bias on this issue center in the area of machine learning, including in binary classification tasks  \cite{chouldechova2017fair}, word embeddings \cite{bolukbasi2016man}, and ranking algorithms \cite{yang2017measuring}. Recent work on social networks also show that recommendation algorithms applied on graphs can exacerbate the bias in the network \cite{stoica2018algorithmic} through a \emph{homophilic} network model. 

Homophily, the tendency of people seeking out to those who are similar to themselves, has been widely observed and studied \cite{lazarsfeld1954friendship}\cite{mcpherson2001birds}. A generative model that combines homophily with rich-get-richer dynamics shows that that homophilous networks can naturally exacerbate the population bias \cite{avin2015homophily}, has later been extended in a series of papers \cite{stoica2018algorithmic}\cite{zhang2021chasm} for better understanding the impact of network structural bias. To our best knowledge, all of the existing work studying the structural bias in social networks solely considers the users' direct network influence; that is, only the first-hop neighbors are considered as the user's network ``social capital''. However, the higher-hop neighbors (indirect influence) can make a big difference in the revenue received by a user through multi-hop referral programs. This paper is the first to study network bias that considers the indirect influence.

\section{Empirical results}
Previous work analyzing structural bias in algorithms applied to social networks primarily focuses on whether the population bias can be amplified when the direct neighborhoods of nodes in the network is considered \cite{avin2015homophily}\cite{stoica2018algorithmic}. Under the schemes of referral programs, we show, among with four network datasets, that considering solely the direct neighbors is not sufficient to conclude an algorithm's impact on a social network's bias in long run, as direction of bias change from 1-hop to 2-hop can be completely opposite to the direction of the bias change from the population ratio to the 1-hop ratio.

\subsection{Datasets}
\subsubsection{Deezer \cite{rozemberczki2020characteristic}}Deezer is a music streaming service. The Deezer dataset was collected from the public API in March 2020. It consists 28,281 nodes that are users from European countries and 92,752 edges that represents the mutual follower relationships between the users. Each node is labelled by the user's gender. Males constitute 44.3\% of the dataset, and are considered as the minority in the network.

\subsubsection{DBLP \cite{stoica2018algorithmic}\cite{ley2009dblp}} The dblp computer science bibliography is a free online database of key computer science journals and proceedings. This dataset contains 102, 263 nodes labelled by gender, and 199, 679 edges. Females constitute 19.3\% of the dataset, and are the minorities in the network.

\subsubsection{Instagram \cite{stoica2018algorithmic}}
Instagram is a photo and video sharing social networking service, where users can follow each other. The dataset contains a total of 553,628 different users whose genders were inferred from their names, and 652,830 edges representing the connections. Females make up 54.4\% of all users in this dataset. Despite the fact that females account for more than half of the data, women are still regarded the disadvantaged for two reasons: (1) other aspects of the network, such as the degree distribution, show a bias against female users; and (2) we wish to remain consistent with previous work using this dataset.

\subsubsection{Twitch \cite{rozemberczki2021twitch}} Twitch is a video live streaming service based in the United States that concentrates on video games. The dataset contains 168,114 nodes that are labelled by the languages used by the users, and 6,797,557 edges representing the mutual follower relationships between them. We focus on the two most frequently used languages, English (majority) and German (minority), and use the sub-graph induced by nodes with the two labels, with 133,610 nodes and 5,762,982 edges.

\subsection{Referral strategies}

Let $\mathcal{G}$ be a social network with a group of minority (referred as ''red'' or $R$) and a group of majority (referred as ''blue'' or $B$). Let $\mathcal{S}$ be a referral strategy applied by members in $\mathcal{G}$. For a $K$-hop referral program with referral strategy $\mathcal{S}$, we denote the \emph{active gain} user $u$ receives from sending $k$-hop referrals actively as $AG_{\mathcal{S}}^{(k)}(u)$, and the \emph{passive gain} $u$ receives from receiving referrals from the $k$th hop as $PG_{\mathcal{S}}^{(k)}(u)$. We assume that the gain (both active and passive) obtained from one successful referral is the same for all referrer-referee pairs. We further denote $AG_{\mathcal{S}}^{(k)}(R)$ ($AG_{\mathcal{S}}^{(k)}(B)$) as the sum of active gains over nodes in red (blue), and $PG_{\mathcal{S}}^{(k)}(R)$ ($PG_{\mathcal{S}}^{(k)}(B)$) as the sum of passive gains over nodes in red (blue). Following conventional graph notations, we denote the number of nodes that are of distance $k$ from $u$ in $\mathcal{G}$ and $\mathcal{S}$ as $deg_{\mathcal{G}}^{(k)}(u)$ and $deg_{\mathcal{S}}^{(k)}(u)$, respectively.

Note that both $AG_{\mathcal{S}}^{(k)}(u)$ and $PG_{\mathcal{S}}^{(k)}(u)$ depend on the referral strategies, as well as the number of successful $k$-hop referrals, which can be impacted by the credibilities of in-between referrers, the price of the product, and the number of $k$-hop friends $u$ has. These factors can be classified into three categories: (1) referral strategy factors, for which referrers select referees according to certain strategies to maximize their rewards, (2) network structure factors, for which a person's position in the network sets caps on the gains this person can get, and (3) referee's decision factor, for which each referral recipient may accept or decline the offer base on their own needs or other criteria. As the third one largely depends on elements that are not closely related to social networks, we focus on the prior two factors. For simplicity, we assume that each user $u$ in the network has an accepting rate $a_u$ after receiving a referral, where $a_u$ is i.i.d. for all the users. The rest of the paper will focus on how referral strategies and network structures change the structural bias against minority nodes.

When designing a referral strategy, referrers may consider elements including the personal relationship with the recipient, the likelihood that for the recipient to accept the referral, and the capabilities for the recipient to promote the referral in the next level. We start the analysis by focusing on the three common motivations in the following simplified \emph{constrained referral strategies}. In each of the referral strategy, we restrict the number of referrals a person can send to a fixed number (for simplicity, we set it to be one) to enforce the impact of referral strategies, and we compare the bias found in the first hop with that in the second hop.
\begin{itemize}
	\item \emph{Random referral strategy} ($\mathcal{S}_R$): each referrer $u$ selects one of the neighbors to share the referral uniformly at random, and the selected referee $v$ accepts the referral with probability $a_v$. Once accepted, $v$ repeats the same selection progress to promote the referral to the second hop.	
 \item \emph{Popularity-driven referral strategy} ($\mathcal{S}_P$): each referral sender $u$ selects the neighbor with the maximum degree among all of $u$'s neighbors to share the referral, and the selected referee $v$ accepts the referral with probability $a_v$. Once accepted, $v$ repeats the same selection progress to promote the referral to the second hop.
	\item \emph{Acceptance-driven referral strategy} ($\mathcal{S}_A$): each referral sender $u$ selects the neighbor with the maximum acceptance rate among all of $u$'s neighbors to share the referral, and the selected referee $v$ accepts the referral with probability $a_v$; if accepted, the referral stops at the first hop; otherwise, the selected referee $v$ passes the referral a neighbor $w$ selected uniformly at random, and the newly selected referee $w$ accepts with probability $a_w$.
\end{itemize}

Finally, to determine how structural bias against minority nodes differ for different values of $k$, we maximize the role a social network can play, and therefore do not limit the number of coupons a user can share. Specifically, we define the following unconstrained \emph{linear referral strategy}:
\begin{itemize}
	\item \emph{Linear referral strategy} ($\mathcal{S}_L$): for each user $u$, we define a random threshold $t_u$, where $t_u$ are i.i.d. for all nodes in $\mathcal{G}$. A referral sender $u$ shares the referral to a neighbor $v$, if and only if $t_u > a_v$. After receiving an referral, the new referrers can select their own referees following the same pattern.
\end{itemize}
Note that, under this referral strategy, different nodes at the same distance $k$ from a node $u$ have the same probability to reward a successful referral to that node $u$, and that this probability remains the same for all $u$. Therefore, to estimate how the expected $k$-hop gains (both active and passive) differ between node $u$ and node $v$, it is sufficient to compare the size of their $k$-hop neighborhood, respectively $deg^{(k)}(u)$ and $deg^{(k)}(v)$.

\subsection{How does bias change from 1-hop to 2-hop?}
Among four networks with binary labelled nodes, we first examine the level of homophily in the networks. Here we define \emph{cross edges} to be the edges with endpoints belonging to two different groups. We compare the fraction of cross edges in the network with the expected fraction of cross edges for the same network with no homophily. Specifically, let $r$ denote the \emph{population ratio} of minority in the network. When there is no homophily in the network, nodes has the equal chance of connecting to nodes of the same group, compared with connecting to nodes of the other group, and the expected fraction of cross-group edges in the network is $2r(1-r)$. We refer to the ratio of the empirical cross edges fraction over the expected fraction as the \emph{homophily rarefaction index}.

We simulate each of the referral strategies. Specifically, we assume $a_u$ and $t_u$ are independent \texttt{Uniform(0,1)} random variables for all $u$ in the network. We then perform each referral strategies 100 times on each of the four networks. In each iteration, we record the ratio of the total number of successful referrals made by red nodes in the 1st hop and in the 2nd hop as active gains and the ratio of the total number of red recipients in the 1st hop and in the 2nd hop as passive gains. The average ratios for active and passive gains are reported in Table \ref{active table} and Table \ref{passive table}, respectively.
\begin{table*}[ht]
\centering
\begin{tabular}{| *{5}{c|} }
	\hline
	 & {Deezer} &{DBLP}  & {Instagram}  & {Twitch} \\
	 \hline
	minority population ratio (\%) & {44.33}& {19.3}   & {54.4} & {7.05}\\
	cross edges (\%)&{47.49}&  {27.5}  & {41.7} &{6.0}\\
	homophily rarefication index & {0.96}&{0.88} & {0.84}  & {0.45}\\
\hline
\hline
	acceptance-driven  -- 1st hop (\%)& {42.20}& {19.09}  & {54.0}&{7.00}  \\
	acceptance-driven -- 2nd hop (\%) &{\color{ao}42.35}& {\color{ao}19.75}  &{\color{ao}55.0} &{\color{ao}7.80}\\
	\hline
	linear -- 1st hop (\%)& 42.85& {18.55}  &{51.64} &{3.81} \\
	linear -- 2nd hop (\%) &{\color{ao}42.91}&{\color{ao}19.81}  &{\color{orange}49.17} &{\color{orange}3.37}\\
	\hline
\end{tabular}
\caption{Active gain: we observe that under the \emph{acceptance-driven referral strategy}, 2-hop active gain always alleviate (colored in green) the bias found in 1-hop; under the \emph{linear referral strategy}, the bias is amplified (colored in orange) in Deezer and DBLP, but reduced in Instagram and Twitch.}
\label{active table}
\end{table*}
\subsubsection{Active gain}
By the construction of \emph{random referral strategy} and \emph{popularity-driven referral strategy}, the active gains for both the one-hop and the two-hop ratios are the network population ratio. Therefore, we may focus solely on the behavior of the \emph{acceptance-driven referral strategy} and the \emph{linear referral strategy}. We observe that, under the \emph{acceptance-driven referral strategy}, the active gain in 2-hop consistently alleviates the bias found in 1-hop in all four datasets, even though the four networks differ in their minority ratio and homophily rarefaction index. 

However, such a consistency is no longer seen in the \emph{linear referral strategy}. We find that networks with different minority population ratio and homophily rarefaction index, behaves differently in 2-hop gains, in terms of reproducing or reducing the bias found in 1-hop degrees. Specifically, for Deezer, in which population differs a little and there is almost no homophily, the bias found in 2-hop degree alleviates the bias in 1-hop active gains; for DBLP, in which the populations of the majority and minority differ a lot, and there is a high homophily rarefication, the bias found in 2-hop active gain alleviates the bias found in 1-hop; for Instagram which has small population difference and relatively high homophily rarefication index, the bias found in 2-hop reinforces the bias found in 1-hop degree; finally, for Twitch, in which there is a large population difference and very low homophily rarefication index, the bias found in 2-hop also reinforces the bias found in 1-hop. As the ratios of the \emph{linear referral strategy} completely replies on the network itself, the above observations indicate that the direction change of structural bias in the network is not consistent. 
\subsubsection{Passive gain}
\begin{table*}[ht]
\centering
\begin{tabular}{| *{5}{c|} }
	\hline
	 & {Deezer} &{DBLP}  & {Instagram}  & {Twitch} \\
	 \hline
	minority population ratio (\%) & {44.33}& {19.28}   & {54.44} & {7.05}\\
	cross edges (\%)&{47.49}&  {27.46}  & {41.67} &{5.93}\\
	homophily rarefication index & {0.96}&{0.88} & {0.84}  & {0.45}\\
\hline
\hline

	random -- 1st hop (\%)& {42.04} & {17.80}  &{51.05}&{5.61} \\
	random -- 2nd hop (\%) &{\color{ao}42.77} & {\color{ao}18.45}  &{\color{ao}52.95} &{\color{ao}5.89}\\
	\hline
	popularity-driven  -- 1st hop (\%) &{38.74} & {16.43}  &{49.80} &{3.00} \\
	popularity-driven -- 2nd hop (\%) &{\color{orange}33.64} & {\color{ao}17.02}  &{\color{orange}40.83} &{\color{orange}0.22}\\
		\hline
	acceptance-driven  -- 1st hop (\%)&{44.30}& {18.10}  &{50.79}&{5.64} \\
	acceptance-driven -- 2nd hop (\%) &{\color{ao}44.52} & {\color{ao}18.74}  &{\color{ao}54.23}&{\color{ao}6.40}\\
	\hline
	linear -- 1st hop (\%)& 42.85& {18.55}  &{51.64} &{3.81} \\
	linear -- 2nd hop (\%) &{\color{ao}42.91}&{\color{ao}19.81}  &{\color{orange}49.17} &{\color{orange}3.37}\\
	\hline
\end{tabular}
\caption{Passive gain: we observe that under the \emph{random referral strategy} and the \emph{acceptance-driven referral strategy}, 2-hop passive gain always alleviate (colored in green) the bias found in 1-hop; under the \emph{popularity-driven referral strategy}, the bias is amplified (colored in orange) for Deezer, DBLP, Instagra, and Twitch, and reduced in DBLP; under the \emph{linear referral strategy}, the ratios for passive gains are the same as the ones for active gains.}
\label{passive table}
\end{table*}
We first note that, under the \emph{linear referral strategy}, the ratios for passive gains are the same as for the ratios for active gains as expected.

For the constrained referral strategies, we observe that the bias found in 1-hop is alleviated in 2-hop, under both the \emph{random referral strategy} and the \emph{acceptance-driven referral strategy}, for all four datasets. On the other hand, under the \emph{popularity-driven referral strategy}, the second hop referral amplifies the bias found in the first hop for Deezer, Instagram, and Twitch, while alleviates it for DBLP. The agreement in the direction of the bias change across different datasets shows that the fairness of referral programs can be systematically affected by the referral strategies employed.

Overall, our observations for the active gain and the passive gain show that the change of bias in multi-hop referral program is complicated, and cannot be simply inferred from the change of bias in the corresponding one-hop referral program. More broadly, these observations alert that for any fairness metric that depends on node degrees, it may not be sufficient to conclude that a network algorithm always amplifies or alleviates the network structural bias by showing only that 1-hop degree behaviors. 

\section{Theoretical analysis for active gains}
We now examine the conditions that lead to an amplification (or alleviation) of the bias found in 1-hop referral programs for active gains. Particularly, we perform mathematical analysis on a homophilous model \cite{avin2015homophily}, show that under the constrained strategies, networks with varying degrees of inherent biases and homophily agree in the direction of bias change from 1-hop to 2-hop, and provide a precise threshold under which the bias in 1-hop gain is amplified in 2-hop gain in the unconstrained linear referral program.

\subsection{Model and definitions}
We first assume two well-observed network growth mechanisms:
\begin{enumerate}
	\item \emph{rich-get-richer} \cite{adamic2000power}\cite{barabasi2002evolution}: existing members with higher degree are likely to attract new connections than the ones with lower degree.
	\item \emph{homophily} \cite{mcpherson2001birds}: members tend to connect with members from the same group.
\end{enumerate}

Formally, for each node in the network, we label it red if belongs to the minority group, and blue otherwise. We use $\mathcal{G}(N_t, t, r, \rho)$ to denote the network generated though the a biased preferrential attachment (BPA) model \cite{avin2015homophily}, where $r$ is the ratio of minority in the population and is less than 0.5, and $\rho$ describes the tendency of a node accepting connections with nodes from the different group. At time $t=1$, we initialize the network with one red node connecting to a blue node. At time $t$, the network grows as follows:
\begin{itemize}
	\item \textbf{Node Growth}: a new node $u$ joins the network, and is colored red with probability $r$ and blue with probability $1-r$.
	\item \textbf{Edge Growth}: 
	\begin{itemize}
	\item the new node $u$ randomly pick an existing node $w$ with a probability proportional to the degree of $w$ in $\mathcal{G}(N_{t-1}, t-1, r, \rho)$ (\emph{rich-get-richer}). 
	\item If the color of $w$ is the same as the color of $u$, then $u$ connects with $w$ directly; otherwise, $u$ accepts the connection with probability $\rho$. If the connection is rejected, then $u$ repicks an existing node until a connection is built successfully (\emph{homophily}).
	\end{itemize}
\end{itemize}

\subsection{The 1-hop ratio is no greater than the minority ratio}
Next, we show that under the four referral programs, the ratios of expected active gain in the  first hop for the red nodes are no greater than the red population ratio.
\begin{lemma}\label{lem_1_hop}
	For a sequence of networks $\{\mathcal{G}(N_t, t, r, \rho)\}$ generated by the BPA model with $r<0.5$, as $t \rightarrow \infty$, under the \emph{random referral strategy} and the \emph{popularity-driven referral strategy}, the ratios of expected 1-hop active gain obtained by the red nodes are the red node population ratio; under the \emph{acceptance-driven referral strategy} and the \emph{linear referral program}, the ratios of expected 1-hop active gain obtained by red nodes are less than the red node population ratio.
\end{lemma}
\begin{proof}
	Denote the expected referral gain obtained by a red (blue) member of degree $k$ as $h(k, r)$ ($h(k, b)$), and the probability of a randomly selected red (blue) node having degree $k$ on $\mathcal{G}(N_t, t, r, \rho)$ as $p^r_t(k)$ ($p^b_t(k)$). Note that all the four referral programs satisfy the following three conditions: (1). $\lim_{t\rightarrow \infty} h(k, i)p_t^i(k) <\infty$, (2). $h(k, b) \geq h(k, r)$, and (3). $h(k,i)$ being non-decreasing for $i\in \{r,b\}$. Furthermore,
\begin{align}
		&\lim_{t\rightarrow \infty}\frac{\sum_{k\geq1}h(k,r)p^r_t(k)}{\sum_{k\geq1}h(k,b)p^b_t(k)} - \lim_{t\rightarrow \infty}\frac{\sum_{k\geq1}p^r_t(k)}{\sum_{k\geq1} p^b_t(k)}\\
		= &\lim_{t\rightarrow \infty} \frac{\sum_{i>j}\left(h(i,b) - h(j,r)\right) (p^r_t(i)p_t^b(j) -  p^r_t(j)p^b_t(i))}{\sum_{k\geq1}h(k,b)p^b_t(k)\sum_{k\geq1}p^b_t(k)}.
		\label{eq-monotone_ineq}
\end{align}
By \cite{avin2015homophily}[Theorem 4.12], we know for all $i>j$, $\lim_{t\rightarrow}p_t^r(k) \propto k^{-\beta (R)}$ and $\lim_{t\rightarrow}p_t^b(k) \propto k^{-\beta (B)}$ with $\beta(R) > \beta(B)$. Therefore,
\begin{equation}
	\lim_{t\rightarrow \infty} \frac{p^r_t(i)}{p^r_t(j)} \leq \lim_{t\rightarrow \infty} \frac{p^b_t(i)}{p^b_t(j)},
\end{equation}
which implies that the limit of (\ref{eq-monotone_ineq}) $\leq 0$. As $\lim_{t\rightarrow \infty}\frac{\sum_{k\geq1}p^r_t(k)}{\sum_{k\geq1} p^b_t(k)} = 1$, we have
\begin{equation}
	\lim_{t\rightarrow \infty}\frac{\sum_{k\geq1}h(k, r)p^r_t(k) }{\sum_{k\geq1}h(k, b)p^b_t(k)} \leq 1.
\end{equation}
The equality holds if and only if $h(k,b) = h(k,r)$ and $h(k,i)$ is constant in $k$, the conditions the \emph{random referral strategy} and the \emph{popularity-driven referral strategy} satisfy.
\end{proof}

\subsection{Constrained referral strategies}
Now we started the analysis of the change of bias from 1-hop to 2-hop for the constrained referral programs.

Because the chance that a red node succeeds in a 2-hop referral is the same as a blue node succeeds under both the \emph{random referral strategy} and the \emph{popularity-driven referral strategy}, the ratios of their expected 2-hop active gains are the population ratios; that is,
	\begin{equation}
		\lim_{t\rightarrow \infty}\frac{\mathbb{E}[AG_{\mathcal{S}_R(t)}^{(2)}(R)]}{\mathbb{E}[AG_{\mathcal{S}_R(t)}^{(2)}(B)]} = \lim_{t\rightarrow \infty}\frac{\mathbb{E}[AG_{\mathcal{S}_P(t)}^{(2)}(R)]}{\mathbb{E}[AG_{\mathcal{S}_P(t)}^{(2)}(B)]}= \frac{r}{1-r}.
	\end{equation}
The following theorem shows that under the \emph{acceptance-driven referral strategy}, the bias found in the first hop cannot be amplified in the second hop.
\begin{corollary}
	Let $\{\mathcal{G}(N_{t-1}, t-1, r, \rho)\}$ be a sequence of graphs generated through the BPA model with $r<1/2$. Under the \emph{acceptance-driven referral strategy} $\mathcal{S}_a$, the ratio of expected active gain for red nodes in the second hop is no greater than that in the first hop. That is,
	\begin{equation}
		\lim_{t\rightarrow \infty}\frac{\mathbb{E}[AG_{\mathcal{S}_A(t)}^{(2)}(R)]}{\mathbb{E}[AG_{\mathcal{S}_A(t)}^{(2)}(B)]} \geq \lim_{t\rightarrow \infty}\frac{\mathbb{E}[AG_{\mathcal{S}_A(t)}^{(1)}(R)]}{\mathbb{E}[AG_{\mathcal{S}_A(t)}^{(1)}(B)]}.
	\end{equation}
\end{corollary}
\begin{proof}
We denote the probability that a referral sent from a red (blue) node being accept at the first hop as $P_r^1$ ($P_b^1$) as ${t\rightarrow \infty}$. As proved in Lemma \ref{lem_1_hop}, they satisfy $P_r^1 \leq P_b^1$. We have the following expression 
\begin{equation}
	\label{eq_ag_1}
	\lim_{t\rightarrow \infty}\frac{\mathbb{E}[AG_{\mathcal{S}_A(t)}^{(1)}(R)]}{\mathbb{E}[AG_{\mathcal{S}_A(t)}^{(1)}(B)]} = \frac{r\cdot P_r^1 }{(1-r)\cdot P_b^1}.
\end{equation}
 For the 2-hop, notice that if the 1-hop friend who receives the referral denies referral and randomly gives the referral to one of its friends, the probability that the referral gets accepted at the 2-hop is actually a constant. We thus have that 
\begin{equation}
	\label{eq_ag_2}
	\lim_{t\rightarrow \infty}\frac{\mathbb{E}[AG_{\mathcal{S}_A(t)}^{(2)}(R)]}{\mathbb{E}[AG_{\mathcal{S}_A(t)}^{(2)}(B)]} = \frac{r\cdot(1-P_r^1) 
	}{(1-r)\cdot (1-P_b^1) 
	},
\end{equation}
Since $P_r^1 \leq P_b^1$, it is easy to check that (\ref{eq_ag_2}) is greater than or equal to (\ref{eq_ag_1}), and the proof is finished.
\end{proof}

\subsection{Linear referral strategy: a precise threshold}
We have seen, in each of the constrained referral strategy, the directions of bias change agree over BPA networks with varying minority ratio and homophily. However, when it comes to the unconstrained linear strategy, the picture is more complex. In this sub-section, we prove mathematically sufficient and necessary conditions determining when 2-hop can effectively mitigate unfairness found in the 1-hop\footnote{All the theorems stated in this subsection on 2-hop active gains stay true when $k$-hop is considered; we present the generalized theorems with proofs for $k$ hop in the appendix. The intuition for the ratio of the expected sum for 2-hop active gains among red nodes being the same as the $k$-hop ratio is that both red nodes and blue nodes extend their $k$-hop neighbors through their direct blue neighbors. Hence, the $k$-hop ratio is essentially the ratio of the number of blue neighbors for red nodes over the number of blue neighbors for blue nodes, for any $k>2$.}. 

We first provide an explicit expression for the ratio of the expected sum for 2-hop active gains among red nodes over that among blue nodes.
\begin{theorem}\label{kth thm} (proof in appendix \ref{ap B})
	For a sequence of networks $\mathcal{G}(N_t, t, r, \rho)$ generated by the BPA model with $r<0.5$, 	\begin{align}
	\lim_{t\rightarrow \infty}\frac{\mathbb{E}[AG_{\mathcal{S}_L(t)}^{(2)}(R)]}{\mathbb{E}[AG_{\mathcal{S}_L(t)}^{(2)}(B)]} =\frac{\beta_2}{ \beta_3},
	\end{align}
where
	\begin{align}
	\beta_2 = \frac{r \rho}{2(1-(1-\alpha^*)(1-\rho))}, \\
	\beta_3 = \frac{1-r}{2(1-\alpha^*(1-\rho))},
\end{align}
and $\alpha^*$ is the unique solution in $[0,1]$ of the equation 
\begin{align}\label{alpha_eqq}
	2 \alpha^{*}=1-(1-r) \frac{\left(1-\alpha^{*}\right)}{1-\alpha^{*}(1-\rho)}+r \frac{\alpha^{*}}{1-\left(1-\alpha^{*}\right)(1-\rho)}.
\end{align}
\end{theorem}
As proven in \cite{avin2015homophily}, $\alpha^*$ in (\ref{alpha_eqq}) is the fraction of the sum of 1-hop active gains among red nodes over that among all nodes, as $t$ goes to infinity, and is always less or equal to the minority population ratio $r$. We show that the proportion of the total 2-hop active gains among red nodes is also no greater than $r$ in the following corollary.
\begin{corollary}\label{2th less than r}(proof in appendix \ref{ap C})
	For a sequence of networks $\mathcal{G}(N_t, t, r, \rho)$ generated by the BPA model with $r<0.5$, under the \emph{linear referral strategy}, as $t\rightarrow \infty$, the ratio of the expected sum for 2-hop active gains among red nodes over that among all nodes is no greater than the minority population ratio $r$. That is,
	\begin{align}\label{thres}
	\frac{\beta_2}{\beta_2 + \beta_3}\leq r,
	\end{align}
where $\beta_2, \beta_3$ are defined as in Theorem \ref{kth thm}.
\end{corollary}

Furthermore, we provide a precise threshold for the amplification (vs. reduction) of bias in 2-hop referral programs.
\begin{corollary}
For a sequence of networks $\mathcal{G}(N_t, t, r, \rho)$ generated by the BPA model, as $t\rightarrow \infty$, the ratio of the expected sum for 2-hop gains among red nodes over that among all nodes is greater than the ratio of the expected sum for 1-hop active gains among red nodes over that among blue nodes if and only if:
	\begin{align}\label{thres}
	\frac{\beta_2}{\beta_2 + \beta_3}	> \alpha^*,
	\end{align}
and $\beta_2, \beta_3$ are defined as in Theorem \ref{kth thm}.
\end{corollary}

Note that we cannot simplify the threshold (\ref{thres}),  into a direct expression on $r$ and $\rho$, as it involves solving $\alpha^*$ which is the root of a cubic equation (so no simple expression for $\alpha^*$). However, with any pairs of fixed $r$ and $\rho$, we can solve $\alpha^*$, and get numerical values for $\beta_2$ and $\beta_3$. In Figure \ref{simulation}, we plot $\beta_2/(\beta_2 + \beta_3)$ and $\alpha^*$ as a function of $\rho$ under four different population ratios.

\begin{figure*}[t]
	\includegraphics[scale =0.35]{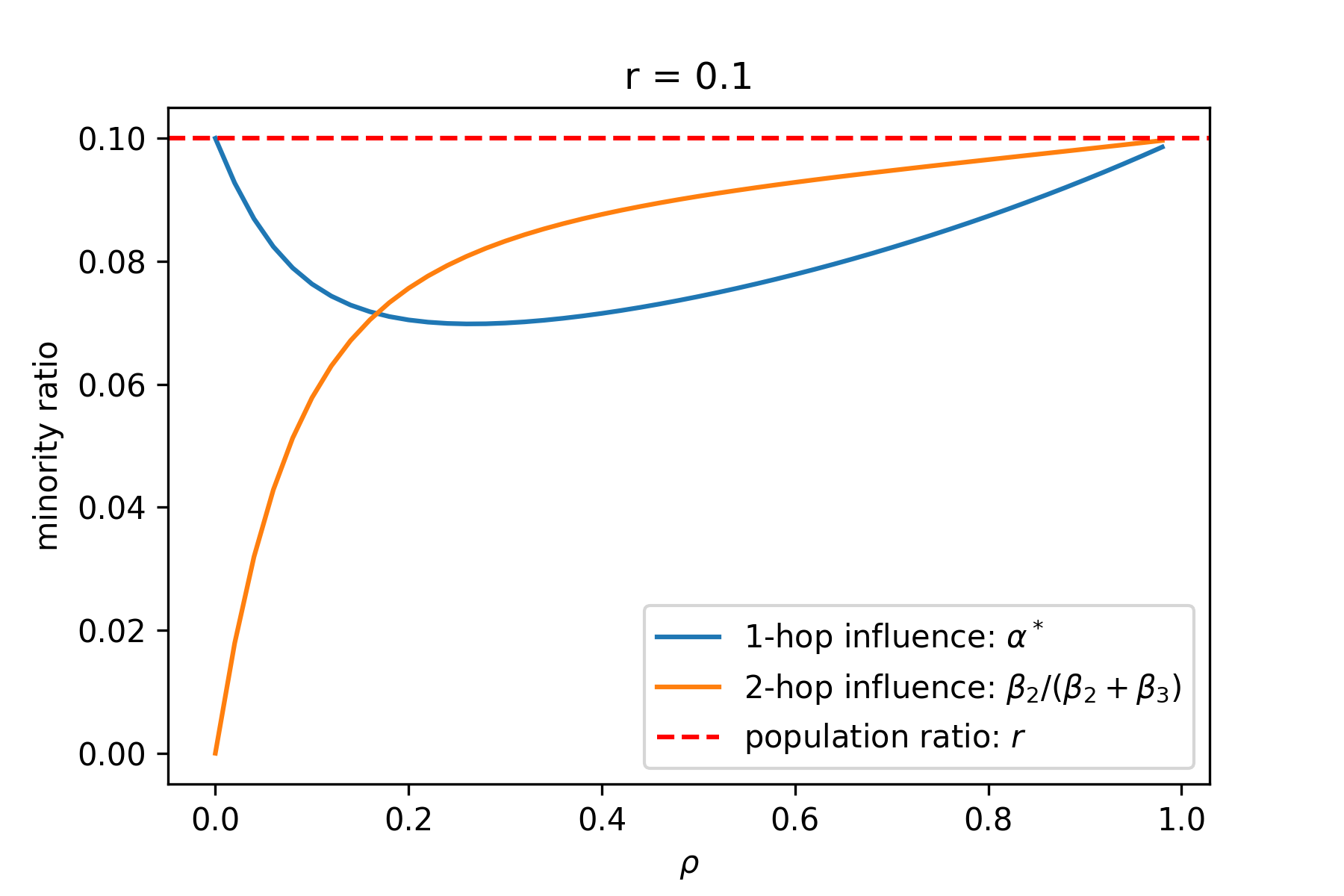}
	\includegraphics[scale =0.35]{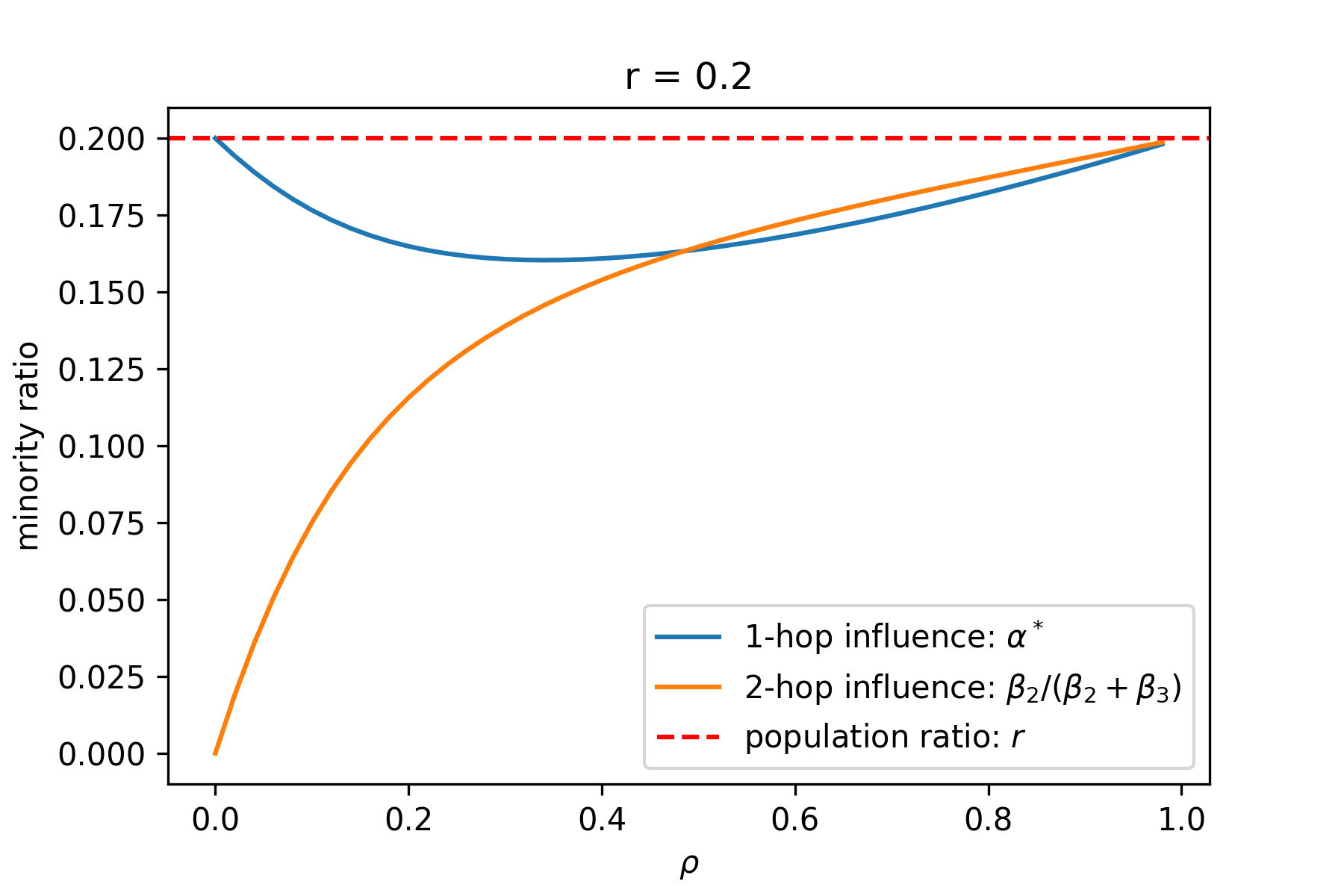}
	\includegraphics[scale =0.35]{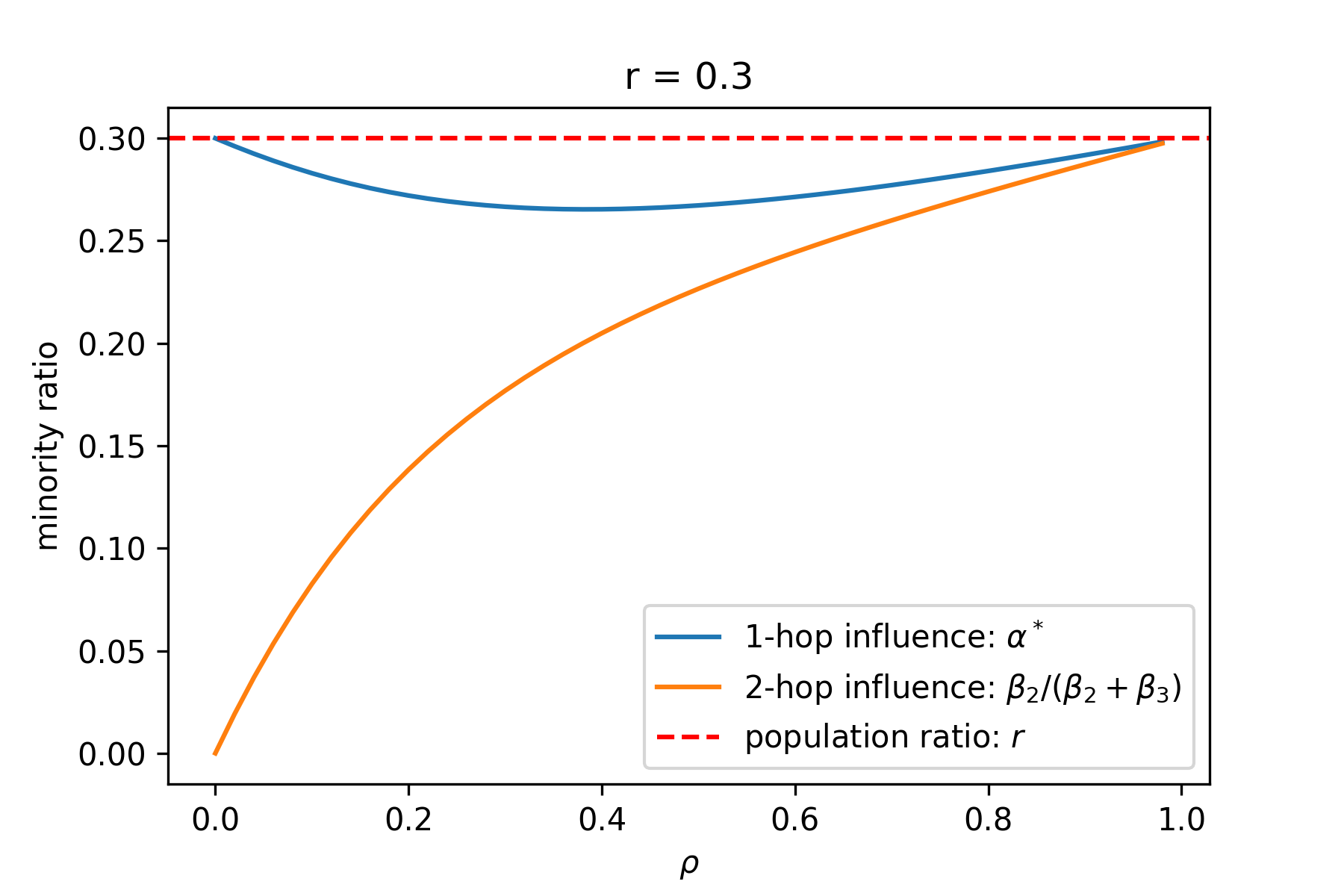}
	\caption{The 2-hop active gains is likely to reduce the 1-hop active gains when (1) the minority ratio in the network is small, and (2) the tendency of connecting with nodes from the other group is large.}
	\label{simulation}
\end{figure*}

We observe that the 2-hop gains reduce the bias found in the 1-hop gains when the population ratio of minorities $r$ is small, and the tendency of accepting a connection from the other group $\rho$ is large. These theoretical results align with our empirical observations in Table \ref{active table} and Table \ref{passive table}. The intuition behind is that red nodes connects to their 2-hop neighbors through their 1-hop neighbors. Among their 1-hop neighbors, the blue neighbors have much higher degrees than the red neighbors. When there is very limited proportion of red nodes in the network and the tendency of red nodes accepting blue connections is large, the chance of a red node connecting to a red node becomes small, which in turn increases the chance of expanding their 2-hop neighbors though blue 1-hop neighbors.

We also notice that when the network is very homophilous, the 1-hop ratios and 2-hop ratios have different limits; that is, as $\rho$ goes to 0, the 1-hop ratio goes to the minority population ratio $r$ as expected in \cite{avin2015homophily}, while the 2-hop ratio goes to $0$. The intuition for the 2-hop ratio is that red nodes connect to their 2-hop neighbors though their direct neighbors. Among the direct neighbors, the blue neighbors have much higher degrees that indeed dominate the count of 2-hop neighbors for red nodes. When $\rho$ is close to zero, red nodes have very limited blue friends, which restrict the red nodes from expanding 2-hop neighbors. We prove this intuition formally in the following corollary.

\begin{corollary}\label{limit} (proof in appendix \ref{ap D})
For a sequence of networks $\mathcal{G}(N_t, t, r, \rho)$ generated by the BPA model with $r<0.5$, as $\rho \rightarrow 0$,
\begin{align}
	\lim_{t\rightarrow \infty}\frac{\mathbb{E}[AG_{\mathcal{S}_L(t)}^{(1)}(R)]}{\mathbb{E}[AG_{\mathcal{S}_L(t)}^{(1)}(B)]}
	 & = \frac{r}{1-r} \text{ and }
	 \lim_{t\rightarrow \infty}\frac{\mathbb{E}[AG_{\mathcal{S}_L(t)}^{(2)}(R)]}{\mathbb{E}[AG_{\mathcal{S}_L(t)}^{(2)}(B)]} = 0.
\end{align}
\end{corollary}
\section{Insights on passive gains}
While active gains represent the direct benefits for nodes sending referrals, people also implicitly benefit from getting access to referrals like commercial coupons and recruitment opportunities. In this section, we extend our analysis on the active gains to passive gains. 

First, note that under the \emph{linear referral strategy}, the ratios in the expected sum of $k$-hop passive gains again are equivalent to the ratios in the expected sum of $k$-hop degrees. Therefore, the conditions under which the bias in 1-hop is alleviated in the second hop for active gains hold for passive gains. 
\begin{corollary}
For a sequence of networks $\mathcal{G}(N_t, t, r, \rho)$ generated by the BPA model, under the \emph{linear referral strategy}, as $t\rightarrow \infty$, the ratio of the expected sum for 2-hop passive gains among red nodes over that among all nodes is greater than the ratio of the expected sum for 1-hop passive gains among red nodes over that among blue nodes if and only if:
	\begin{align}\label{thres}
	\frac{\beta_2}{\beta_2 + \beta_3}	> \alpha^*,
	\end{align}
and $\beta_2, \beta_3$ are defined as in Theorem \ref{kth thm}.
\end{corollary}

Now, to analyze the behavior for the constrained referral strategies, we simulate 100 BPA networks on 10,000 nodes for different sets of population ratios $r$ and homophilous parameters $\rho$. On each network, we perform the three constrained referral strategies. We then report the average red passive gain ratios for the first hop and for the second hop in Figure \ref{passive simulation}. 

\begin{figure*}[h]
\includegraphics[height = 3.5cm]{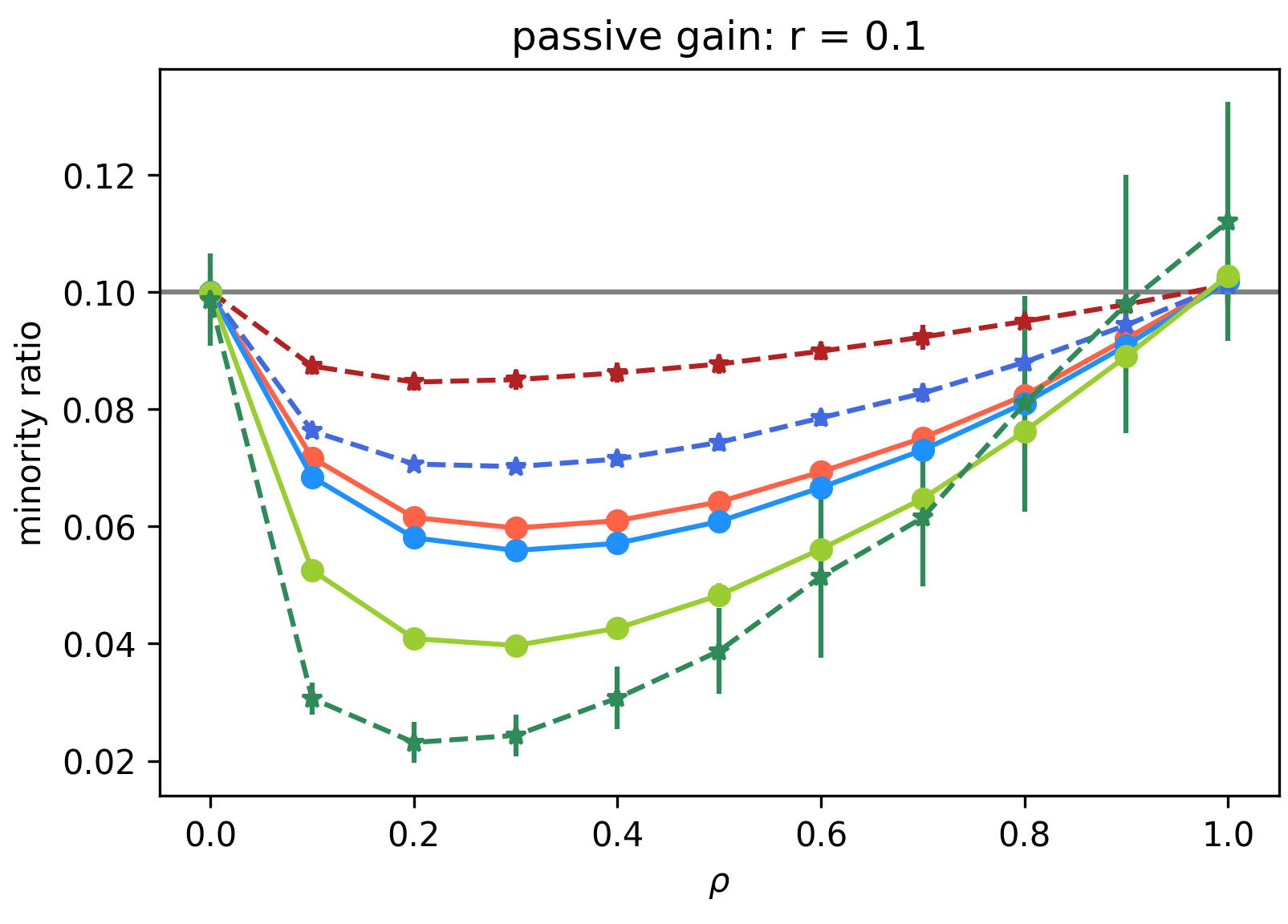}
\includegraphics[height = 3.5cm]{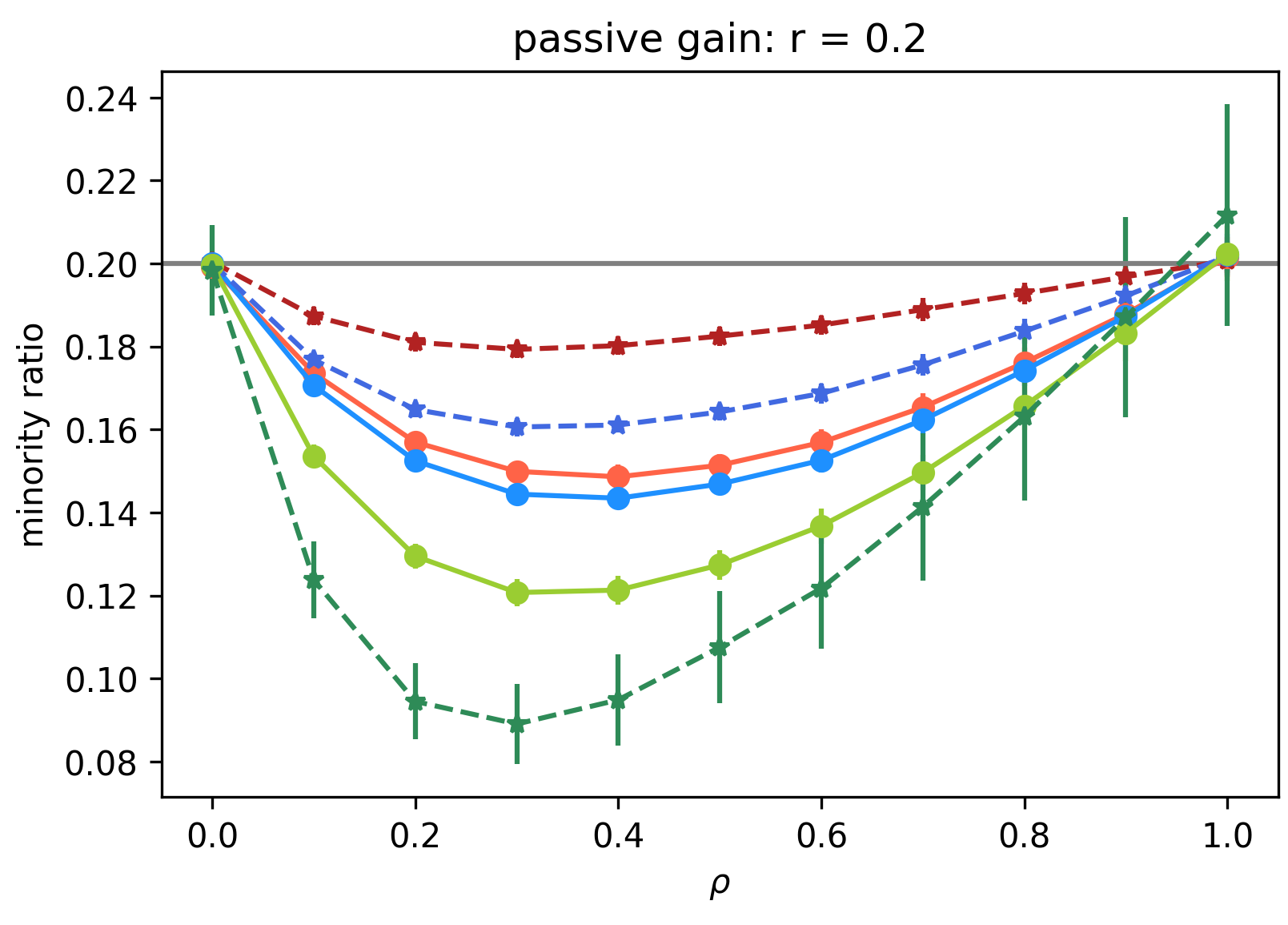}
\includegraphics[height = 3.5cm]{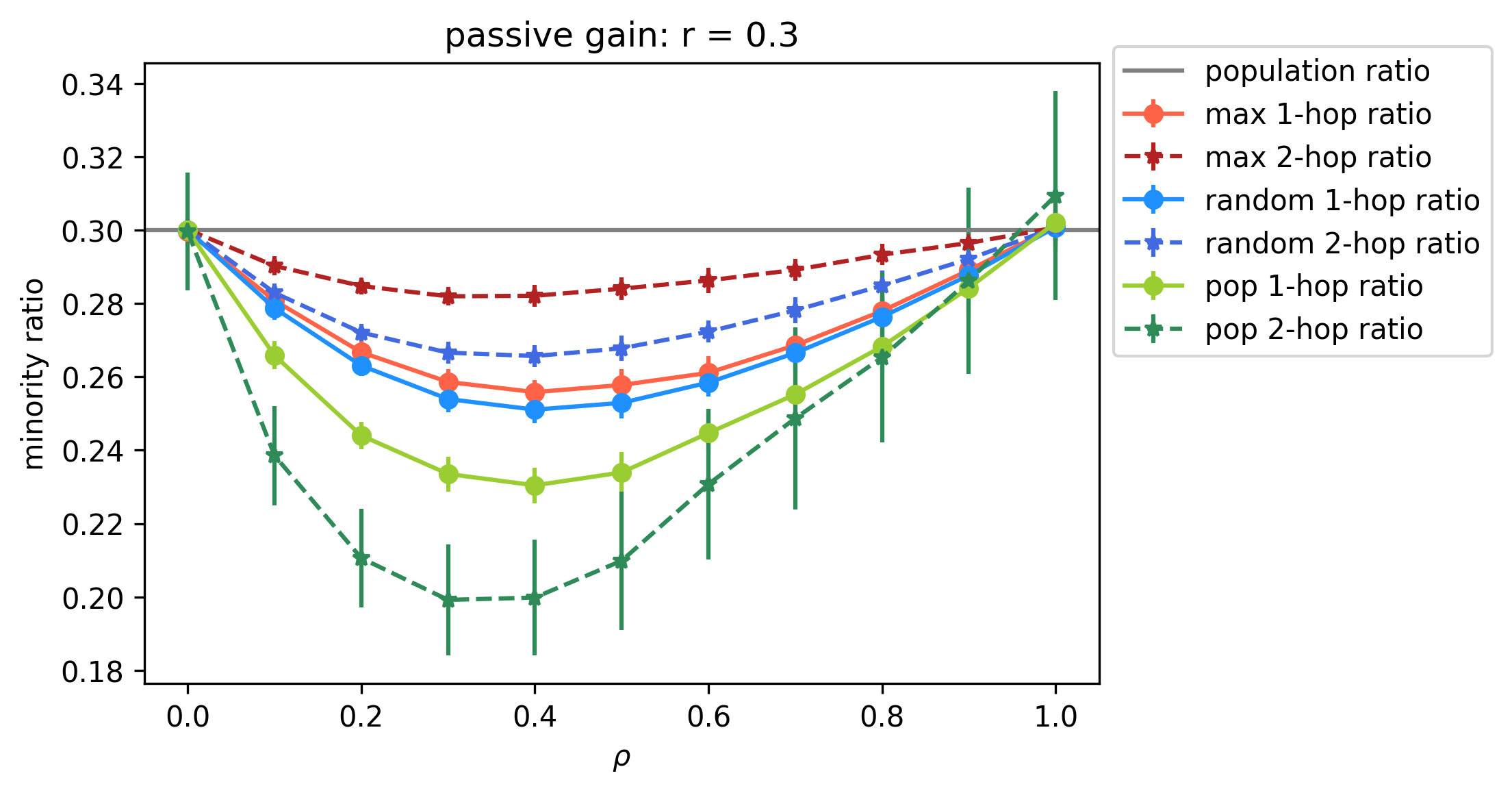}
	\caption{Under the \emph{random referral strategy}, the bias of 1-hop passive gain is alleviated in the second hop; under the \emph{popularity-driven referral strategy}, the bias of 1-hop passive gain is amplified in the second hop; under the \emph{acceptance-driven referral strategy}, the bias of 1-hop passive gain is amplified in the second hop.}
	\label{passive simulation}
\end{figure*}

We observe that unlike the linear referral program where the direction of bias change from 1-hop to 2-hop depends on the minority ratio and homophily, under each constrained referral strategy, the direction across different variables remain consistent. Indeed, as $\rho$ goes to $0$ or $1$, both the ratios for the expected sum of red passive gain in 1-hop and in 2-hop converge to the minority population ratio, which makes the two ratio trends over various $\rho$ less likely to cross. Formally,
\begin{lemma}
For a sequence of networks $\mathcal{G}(N_t, t, r, \rho)$ generated by the BPA model with $r<0.5$, under any of the three constrained referral strategies, as $t\rightarrow \infty$, as $\rho \rightarrow 0$ or $\rho \rightarrow 1$,
	\begin{align}
	\lim_{t\rightarrow \infty}\frac{\mathbb{E}[PG^{(2)}_{\mathcal{S}(t)}(R)]}{\mathbb{E}[PG^{(2)}_{\mathcal{S}(t)}(B)]} = \lim_{t\rightarrow \infty}\frac{\mathbb{E}[PG^{(1)}_{\mathcal{S}(t)}(R)]}{\mathbb{E}[PG^{(1)}_{\mathcal{S}(t)}(B)]}
	 =  \frac{r}{1-r},
\end{align}
\end{lemma}
We now attempt to understand the direction of bias change for each of the three constrained referral strategies. To do that, we introduce two well-known network effects: \emph{friendship paradox} and \emph{glass-ceiling effect}.
\begin{lemma}\label{paradox}
	(friendship Paradox in BPA) Let $\{\mathcal{G}(N_{t-1}, t-1, r, \rho)\}$ be a sequence of graphs generated through the BPA model with $r<1/2$. For $c_1, c_2 \in \{R,B\}$, denote $F_{c_1c_2}(t)$ as the degree of the $c_2$ color node of a randomly chosen edge that connects two nodes with $c_1, c_2$ colors (if $c_1 = c_2$, randomly choose one end of the edge). We then have
	\begin{equation}
		\lim_{t \rightarrow \infty} \mathbb{E}[F_{c_1c_2}(t)] \geq 
		\lim_{t \rightarrow \infty} \mathbb{E}[d_{c_2}(t)],
	\end{equation}
	where $d_{c_2}(t)$ is the degree of a randomly chosen $c_2$ color node.
\end{lemma}
\begin{lemma}(Theorem 4.1 \cite{avin2015homophily}, glass-ceiling effect)
	Let $\{\mathcal{G}(N_{t-1}, t-1, r, \rho)\}$ be a sequence of graphs generated through the BPA model with $r<1/2$. Let $\text{top}_k(R)$ ($\text{top}_k(B)$) denote the number of red (blue) nodes that have degree at least $k$ in a network. There exists an increasing function $k(n)$ such that $\lim_{t\rightarrow\infty} \text{top}_k(B) = \infty$ and 
	\begin{equation}
		\lim_{t\rightarrow \infty} \frac{\text{top}_k(R)}{\text{top}_k(B)} = 0.
	\end{equation}
\end{lemma}

Now, for each of the constrained referral strategies, we provide our analysis below:

\textbf{random referral strategy:} The pair of blue lines in Figure \ref{passive simulation} shows that the bias found in the first hop for the passive gain over red nodes is always alleviated in the second hop. Intuitively, this phenomenon can be explained by the fact that the random selection is color-blind. Thus, the first hop selection alleviate the effect of homophily in the next round of selection, and therefore alleviate the bias found in the first hop.

\textbf{popularity-driven referral strategy:}
Through the pair of green lines of Figure \ref{passive simulation}, we observe that the bias found in the first hop for the passive gain over red nodes is always amplified in the second hop. A potential explanation is that: the intermediate referrer' degree can be large, not only by the Friendship Paradox, but also because the intermediate referrer has the highest degree among the neighbors of the initial referrer. Therefore, the intermediate referrers are more likely to reach to the most popular nodes in the network than the initial referrers. Because the glass-ceiling effect indicates that the most popular nodes in the network are mostly blue, the second hop reinforces the bias found in the first hop.

\textbf{acceptance-driven referral strategy:} we see in the pair of red lines of Figure \ref{passive simulation} that the bias in the second hop are always better than that in the second hop. Under the \emph{acceptance-driven referral strategy}, we see there are two counteracting factors that can affect the ratio for the second hop. First, because red nodes have lower degrees than blue nodes in average, the chance that a referral sent by a red node is accepted at the first hop is lower than a referral sent by a blue node. This gives red nodes larger opportunities to enter the second hop. On the other hand, as blue initiators are more likely to choose blue intermediate referrers, once entered the second hop, blue nodes are more likely to succeed in sending the referral to the new referees. When the former factor overcomes the later factor, the bis found in 1-hop would be released in the 2-hop.
\section{Conclusion}
The recent web-based form referral programs have made today's product promotion easier than ever. At the same time, the historical bias in the population evolves quickly with the simplest network primitives as well as complex algorithmic rules. As the first to examine the fairness in multi-hop referral programs, our empirical results suggest and theory results confirm, leveraging higher-hop referral programs can amplify bias or, more interestingly, contribute to fairness. It primarily depends on which referal strategy users employ and whether active or passive gains are considered. For constrained strategies, networks with varying degrees of inherent biases and homophily agree in their direction of bias change. When removing strategic effects and facing the referral programs directly to the network structual bias, we find that such an agreement no long holds, and we prove a necessary and sufficient condition under which the bias found in 1-hop is alleviated in the higher hops for such referral programs.

This observation is critical, as previous literature analyzing the impact of algorithms on the bias of social networks primarily focuses on users' 1-hop neighbors, while our results indicate that understanding solely the 1-hop behavior is not sufficient to conclude the impact social network algorithms have in long run. This paper makes the first attempt to understand the bias in higher-hop neighbors. 

We hope our work can serve as a starting point for the bias analysis in multi-hop referral programs, as well as in long-run analysis for other network algorithms.

\bibliographystyle{unsrt}
\bibliography{draft2111.bib}
\appendix

\section{Notations}
For clarity, we list all notations that are used in our theory presentation in Table \ref{notations}.

\begin{table*}
    \begin{tabularx}{\textwidth}{p{0.22\textwidth}X}
    \toprule
      $V_C(t)$ & the set of nodes in color $C$ at time $t$;\\
      $d^{(k)}_i(t)$ & the $k$-hop degree of node $i$ at time $t$;\\
      $d^{(k)}_{i,C}(t)$ & the number of nodes in color $C$ among the $k$-hop neighbors of node $i$ at time $t$;\\
   	  $d^{(k)}_{i,C*}(t)$ & among the $k$-hop neighbors of node $i$ at time $t$, the number of nodes from which the second node on the unique path to $i$ is of color $C$;\\
   	  $D^{(k)}_C(t)$ & sum of $k$-hop degrees over nodes of color $C$ at time $t$; that is, $\sum_{i \in V_C(t)}d^{(k)}_i(t)$;\\
   	   $\alpha(t)$ & fraction of the sum of 1-hop degree for red nodes over that for all nodes; that is, $D^{(1)}_R(t)/(2t)$;\\
   	   $M^{(k)}_{C_1,**}(t)$ & $\sum_{i \in V_{C_1}(t)}d^{(1)}_i(t)d^{(k)}_i(t)$;\\
   	   $M^{(k)}_{C_1,**}(t)$ & $\sum_{i \in V_{C_1}(t)}d^{(1)}_i(t)d^{(k)}_i(t)$;\\
   	 $M^{(k)}_{C_1,*C_2}(t)$ & $\sum_{i \in V_{C_1}(t)}d^{(1)}_i(t)d^{(k)}_{i,C_2}(t)$; \\
		$M^{(k)}_{C_1,C_2*}(t)$ & $\sum_{i \in V_{C_1}(t)}d^{(1)}_i(t)d^{(k)}_{i,C_2*}(t)$ with $M^{(1)}_{C_1,C_2*}(t):= M^{(1)}_{C_1, **}(t)$.
   	 \\
    \bottomrule
     \end{tabularx}
     \caption{Notation}\label{notations}
     \vspace{-7mm}
    \end{table*}

\section{Proof of Theorem \ref{kth thm}}\label{ap B}
\begin{reptheorem}{kth thm}[extended]
	For a sequence of networks $\mathcal{G}(N_t, t, r, \rho)$ generated by the BPA model with $r<0.5$, 	for any $k>1$,
	\begin{align}
	\lim_{t\rightarrow \infty}\frac{\mathbb{E}[AG_{\mathcal{S}_L(t)}^{(k)}(R)]}{\mathbb{E}[AG_{\mathcal{S}_L(t)}^{(k)}(B)]} =\frac{\beta_2}{ \beta_3},
	\end{align}
where
	\begin{align}
	\beta_2 = \frac{r \rho}{2(1-(1-\alpha^*)(1-\rho))}, \\
	\beta_3 = \frac{1-r}{2(1-\alpha^*(1-\rho))},
\end{align}
and $\alpha^*$ is the unique solution in $[0,1]$ of the equation 
\begin{align}\label{alpha_eqq}
	2 \alpha^{*}=1-(1-r) \frac{\left(1-\alpha^{*}\right)}{1-\alpha^{*}(1-\rho)}+r \frac{\alpha^{*}}{1-\left(1-\alpha^{*}\right)(1-\rho)}.
\end{align}
\end{reptheorem}
\begin{proof}
	
For sequences $\{a_t, t \geq 1\}$, $\{b_t, t \geq 1\}$ and $\{c_t, t \geq 1\}$, we denote by $c_t = o_t(1)$ if $\lim_{t \rightarrow \infty} c_t = 0$; we denote by $a_t = O(b_t)$ if there exists some $C \neq 0$ such that $a_t/b_t - C = o_t(1)$.

For the growth model, at each time $t$, there are 4 possibilities: 
	\begin{itemize}
		\item case 1: a new red node connects a red node;
		\item case 2: a new red node connects a blue node;
		\item case 3: a new blue node connects a blue node;
		\item case 4: a new blue node connects a red node. 
	\end{itemize}

In BPA model, the probability for above cases are:
\begin{align}
	\mathbb{P}(\text{case 1 at time t}) &= \frac{r \alpha(t)}{2(1-(1-\alpha(t))(1-\rho))}, \\
	\mathbb{P}(\text{case 2 at time t}) &= \frac{r (1-\alpha(t)) \rho}{2(1-(1-\alpha(t))(1-\rho))}, \\
	\mathbb{P}(\text{case 3 at time t}) &= \frac{(1-r) (1-\alpha(t))}{2(1-\alpha(t)(1-\rho))}, \\
	\mathbb{P}(\text{case 4 at time t}) &= \frac{(1-r) \alpha(t) \rho}{2(1-\alpha(t)(1-\rho))}.
\end{align}

We can write a recursion for $D^{(k)}_C(t)$ by considering those 4 cases. Denote by $\{\mathcal{F}_t, t \geq 1\}$ the filtration such that $\mathcal{F}_t$ contains the information of the graph at time $t$. We have  
\begin{align}
	\label{eq_Dk_recur_1}
	&\mathbb{E}[D^{(k)}_R(t+1) \mid \mathcal{F}_t]
	= D^{(k)}_R(t) + \beta_1(t) \frac{M^{(k-1)}_{R,**}(t) + M^{(k-1)}_{R,*R}(t)}{t}\\
	&+ \beta_2(t) \frac{M^{(k-1)}_{B,**}(t) + M^{(k-1)}_{B,*R}(t)}{t}  + \beta_3(t) \frac{M^{(k-1)}_{B,*R}(t)}{t} + \beta_4(t) \frac{M^{(k-1)}_{R,*R}(t)}{t},
\end{align}
where
\begin{align}
	\beta_1(t) = \frac{r}{2(1-(1-\alpha(t))(1-\rho))}, \\
	\beta_2(t) = \frac{r \rho}{2(1-(1-\alpha(t))(1-\rho))}, \\
	\beta_3(t) = \frac{1-r}{2(1-\alpha(t)(1-\rho))}, \\
	\beta_4(t) = \frac{(1-r) \rho}{2(1-\alpha(t)(1-\rho))}.
\end{align}
From \cite{avin2015homophily}, we see that $\alpha(t) \rightarrow \alpha^*$ almost surely, where $\alpha^*$ satisfies the following equation.
\begin{equation}
	2 \alpha^* = \frac{2 r \alpha^* + r (1-\alpha^*) \rho}{1 - (1-\alpha^*) (1-\rho)}
				+ \frac{ (1-r) \alpha^* \rho}{1 - \alpha^* (1-\rho)}.
\end{equation}
Hence we also have that $\beta_i(t) \rightarrow \beta_i$ almost surely for $i = 1,2,3,4$, where

\begin{align}
	\beta_1 = \frac{r}{2(1-(1-\alpha^*)(1-\rho))}, \\
	\beta_2 = \frac{r \rho}{2(1-(1-\alpha^*)(1-\rho))}, \\
	\beta_3 = \frac{1-r}{2(1-\alpha^*(1-\rho))}, \\
	\beta_4 = \frac{(1-r) \rho}{2(1-\alpha^*(1-\rho))}.
\end{align}

Taking expectations on both sides of (\ref{eq_Dk_recur_1}), we get

\begin{align}
	\label{eq_Dk_recur_1}
	&\mathbb{E}[D^{(k)}_R(t+1)]
	= \mathbb{E}[D^{(k)}_R(t)] + \mathbb{E}\left[\beta_1(t) \frac{M^{(k-1)}_{R,**}(t) + M^{(k-1)}_{R,*R}(t)}{t}\right]\\
	&+ \mathbb{E}\left[\beta_2(t) \frac{M^{(k-1)}_{B,**}(t) + M^{(k-1)}_{B,*R}(t)}{t} \right] + \mathbb{E}\left[\beta_3(t) \frac{M^{(k-1)}_{B,*R}(t)}{t}\right] \\
	&+ \mathbb{E}\left[\beta_4(t) \frac{M^{(k-1)}_{R,*R}(t)}{t}\right].
\end{align}
Similarly,
\begin{align}
	&\mathbb{E}[D^{(k)}_B(t+1)]
	= \mathbb{E}[D^{(k)}_B(t)] + \mathbb{E}\left[\beta_3(t) \frac{M^{(k-1)}_{B,**}(t) + M^{(k-1)}_{B,*B}(t)}{t}\right]\\
	&+ \mathbb{E}\left[\beta_4(t) \frac{M^{(k-1)}_{R,**}(t) + M^{(k-1)}_{R,*B}(t)}{t} \right] 
	+ \mathbb{E}\left[\beta_1(t) \frac{M^{(k-1)}_{R,*B}(t)}{t} \right] \\
	&+ \mathbb{E}\left[ \beta_2(t) \frac{M^{(k-1)}_{B,*B}(t)}{t} \right].
\end{align}

We have the following lemma.

\begin{lemma}
	\label{lem_exp_ratio}
	For any color $C \in \{R, B\}$, any $i \in \{1,2,3,4\}$ and any $k$, we have that
	\begin{equation}
		\frac{\mathbb{E}\left[ \beta_i(t) M^{(k-1)}_{C,**}(t) \right]}
		{\beta_i \mathbb{E}\left[ M^{(k-1)}_{C,**}(t) \right]}
		= 1 + o_t(t^{-1/4}).
	\end{equation}
\end{lemma}

Using the above lemma, we get
\begin{align}
&	\mathbb{E}[D^{(k)}_R(t+1)]
	= \mathbb{E}[D^{(k)}_R(t)] + \beta_1 (1 + o_t(1))\frac{\mathbb{E}\left[M^{(k-1)}_{R,**}(t)\right] + \mathbb{E}\left[M^{(k-1)}_{R,*R}(t)\right]}{t}\\
	&+ \beta_2 (1 + o_t(1))\frac{\mathbb{E}\left[M^{(k-1)}_{B,**}(t)\right] + \mathbb{E}\left[M^{(k-1)}_{B,*R}(t)\right]}{t}  \\
	&+  \beta_3(1 + o_t(1))\frac{\mathbb{E}\left[M^{(k-1)}_{B,*R}(t)\right]}{t} + \beta_4 (1 + o_t(1))\frac{\mathbb{E}\left[M^{(k-1)}_{R,*R}(t)\right]}{t}.
\end{align}

\begin{align}
	&\mathbb{E}[D^{(k)}_B(t+1)]
	= \mathbb{E}[D^{(k)}_B(t)] + \beta_3 (1 + o_t(1)) \frac{\mathbb{E}\left[M^{(k-1)}_{B,**}(t)\right] + \mathbb{E}\left[ M^{(k-1)}_{B,*B}(t)\right]}{t} \\
	&+ \beta_4 (1 + o_t(1))\frac{\mathbb{E}\left[M^{(k-1)}_{R,**}(t)\right] + \mathbb{E}\left[M^{(k-1)}_{R,*B}(t)\right]}{t}  \\
	& + \beta_1 (1 + o_t(1)) \frac{\mathbb{E}\left[M^{(k-1)}_{R,*B}(t)\right]}{t}  + 
	 \beta_2 (1 + o_t(1)) \frac{\mathbb{E}\left[M^{(k-1)}_{B,*B}(t)\right]}{t}.
\end{align}

Assume that $r < 1/2$. We need the following lemmas.
\begin{lemma}
	\label{lem-M-order}
	We have that, for $k \geq 1$,
	\begin{align}
		\mathbb{E} [M^{(k)}_{B,**}(t) ] &= O((\log t)^{k-1} t^{2 \beta_2 + 2 \beta_3}), \\
		\mathbb{E} [M^{(k)}_{B,B*}(t) ] &= \mathbb{E} [M^{(k)}_{B,**}(t) ] (1+o_t(1)),\\
		\mathbb{E} [M^{(k)}_{B,*B}(t) ] &= \frac{\beta_3}{\beta_2 + \beta_3} \mathbb{E} [M^{(k)}_{B,**}(t) (1+o_t(1)) ],\\
		\mathbb{E} [M^{(k)}_{R,**}(t) ] &= o_t(\mathbb{E} [M^{(k)}_{B,**}(t) ]).
	\end{align}
\end{lemma}

\begin{lemma}
	\label{lem-recursion-order}
	For a sequence $\{a_t\}$ satisfying
	\begin{equation}
		a_{t+1} = a_t + c_1 d_t \frac{a_t}{t} + c_2 b_t, 
	\end{equation}
	where $d_t = (1+o_t(t^{-\epsilon}))$ for some $\epsilon > 0$, $b_t = (\log t)^m t^{c_3 - 1} (1+o_t(1))$. 
	We then have that
	\begin{equation}
		a_t = 
		\begin{cases}
			O(t^{c_1}), &\text{if $c_1 > c_3$} \\
			O((\log t)^{m+1} t^{c_3}), &\text{if $c_1 = c_3$} \\
			\frac{c_2}{c_3-c_1} (\log t)^m t^{c_3} (1+o_t(1)), &if c_1 < c_3.
		\end{cases}
	\end{equation}
\end{lemma}

By Lemma \ref{lem-M-order}, we see that
\begin{align}
	\mathbb{E}[D^{(k)}_R(t+1)]
	&= \mathbb{E}[D^{(k)}_R(t)] + 2 \beta_2 \frac{\mathbb{E}[M^{(k-1)}_{B,**}(t)]}{t} (1+o_t(1)).
\end{align}
Similarly, we have
\begin{align}
	\mathbb{E}[D^{(k)}_B(t+1)]
	&= \mathbb{E}[D^{(k)}_B(t)] + 2 \beta_3 \frac{\mathbb{E}[M^{(k-1)}_{B,**}(t)]}{t} (1+o_t(1)).
\end{align}
By Lemma \ref{lem-recursion-order} we have that
\begin{align}
	\frac{\mathbb{E}[D^{(k)}_R(t)]}{\mathbb{E}[D^{(k)}_B(t)]}
	 = \frac{\beta_2}{\beta_3} (1+o_t(1)).
\end{align}
\end{proof}

\section{Proof of Corollary \ref{2th less than r}}\label{ap C}
\begin{repcorollary}{2th less than r}[extended]
For a sequence of networks $\mathcal{G}(N_t, t, r, \rho)$ generated by the BPA model with $r<0.5$, under the \emph{linear referral strategy}, as $t\rightarrow \infty$, the ratio of the expected sum for k-hop degrees among red nodes over that among all nodes is no greater the minority population ratio $r$. That is,
	\begin{align}
	\frac{\beta_2}{\beta_2 + \beta_3} \leq r,
	\end{align}
where $\beta_2, \beta_3$ are defined as in Theorem \ref{kth thm}.
\end{repcorollary}
\begin{proof}
By the definition of $\beta_2$ and $\beta_3$, we know
\begin{align}
	\frac{\beta_2}{\beta_2 + \beta_3} \leq r & \Leftrightarrow \frac{\rho}{1-(1-\alpha^*)(1-\rho)} \leq \frac{1}{1-\alpha^*(1-\rho)}\\
	& \Leftrightarrow \rho -\alpha^* \rho(1-\rho) \leq 1- (1-\alpha^*)(1-\rho)\\
	& \Leftrightarrow
	    1- \rho - (1-\rho)(1-\alpha^* - \alpha^* \rho) 
	    \geq 0\\
	& \Leftrightarrow \alpha^* (1+\rho)(1-\rho) \geq 0,
\end{align}
which is clearly true by the definition of $\alpha^*$ and $\rho$.
\end{proof}

\section{Proof of Corollary \ref{limit}}\label{ap D}
\begin{repcorollary}{limit}[extended]
For a sequence of networks $\mathcal{G}(N_t, t, r, \rho)$ generated by the BPA model with $r<0.5$, as $\rho \rightarrow 0$,
\begin{align}
	\lim_{t\rightarrow \infty}\frac{\mathbb{E}[AG_{\mathcal{S}_L(t)}^{(1)}(R)]}{\mathbb{E}[AG_{\mathcal{S}_L(t)}^{(1)}(B)]}
	 & = \frac{r}{1-r} \text{ and }
	 \lim_{t\rightarrow \infty}\frac{\mathbb{E}[AG_{\mathcal{S}_L(t)}^{(k)}(R)]}{\mathbb{E}[AG_{\mathcal{S}_L(t)}^{(k)}(B)]} = 0.
\end{align}
\end{repcorollary}
\begin{proof}
From \cite{avin2015homophily}, we know that as $t\rightarrow \infty$, $\frac{\mathbb{E}[D^{(1)}_R(t)]}{\mathbb{E}[D^{(1)}_B(t)]} = \frac{\alpha^*}{(1-\alpha^*)}$, where 
\begin{align}
	2 \alpha^* = \frac{2 r \alpha^* + r (1-\alpha^*) \rho}{1 - (1-\alpha^*) (1-\rho)}
				+ \frac{ (1-r) \alpha^* \rho}{1 - \alpha^* (1-\rho)}.
\end{align}
When $\rho \rightarrow 0$, $\alpha^* \rightarrow r$, and therefore $\frac{\mathbb{E}[D^{(1)}_R(t)]}{\mathbb{E}[D^{(1)}_B(t)]}
	 \rightarrow \frac{r}{1-r}$.

	For $k\geq 2$, we know from Theorem \ref{kth thm} that, as $t\rightarrow \infty$, 
	\begin{align}
		\lim_{\rho \rightarrow 0} \frac{\mathbb{E}[D^{(k)}_R(t)]}{\mathbb{E}[D^{(k)}_B(t)]}  =\lim_{\rho \rightarrow 0} \frac{\beta_2}{\beta_3} = \lim_{\rho \rightarrow 0} \frac{r\rho}{1-r}\cdot 
		\frac{2(1-\alpha^*(1-\rho))}{2(1-(1-\alpha^*)(1-\rho))}=0.
	\end{align}

\end{proof}

\section{Proof of lemmas}

\subsection{Proof of Lemma \ref{lem-M-order}}

\begin{replemma}{lem-M-order}
	We have that, for $k \geq 1$,
	\begin{align}
		\mathbb{E} [M^{(k)}_{B,**}(t) ] &= O((\log t)^{k-1} t^{2 \beta_2 + 2 \beta_3}), \label{B,**}\\
		\mathbb{E} [M^{(k)}_{B,B*}(t) ] &= \mathbb{E} [M^{(k)}_{B,**}(t) ] (1+o_t(1)), \label{B,B*}\\
		\mathbb{E} [M^{(k)}_{B,*B}(t) ] &= \frac{\beta_3}{\beta_2 + \beta_3} \mathbb{E} [M^{(k)}_{B,**}(t) (1+o_t(1)) ], \label{B,*B}\\
		\mathbb{E} [M^{(k)}_{R,**}(t) ] &= o_t(\mathbb{E} [M^{(k)}_{B,**}(t) ])\label{R,**}.
	\end{align}
\end{replemma}

\begin{proof}

We prove it by induction. Recall that $M^{(1)}_{B,B*}(t) = M^{(1)}_{B,**}(t)$ by definition, and thus (\ref{B,B*}) trivially holds. We first show that the claims hold for $k=1$. First note that $M_{B, **}^{(1)}(t+1)$ arises from $M_{B, **}^{(1)}(t)$ in the following cases:
\begin{enumerate}
	\item A new red node connects with an existing blue node $i$: node $i$ then contributes $\left(d_i^{(1)}(t) + 1\right)^2 - \left(d_i^{(1)}(t)\right)^2 = 1 + 2d_i(t)$ increments. 
	\item A new blue node connects with an existing red node $i$: the new node contributes 1 increment.
	\item A new blue node connects with an existing blue node $i$: node $i$ then contributes $\left(d_i^{(1)}(t) + 1\right)^2 - \left(d_i^{(1)}(t)\right)^2 = 1 + 2d_i(t)$ increments, and the the new node contributes 1 increment.
\end{enumerate}

Therefore,
\begin{align}
	\mathbb{E}[M^{(1)}_{B,**}(t+1) \mid \mathcal{F}_t]
	&= M^{(1)}_{B,**}(t) + 
	\beta_2(t) \frac{D_B^{(1)}(t) + 2 M^{(1)}_{B,**}(t)}{t}\\
	&  + \beta_3(t)\left(1 +\frac{D_B^{(1)}(t) + 2 M^{(1)}_{B,**}(t)}{t}\right) + \beta_4(t)
\end{align}
Taking expectations on both sides gives
\begin{align}
	\mathbb{E}[M^{(1)}_{B,**}(t+1)]
	&= \mathbb{E}[M^{(1)}_{B,**}(t)] + 
	\mathbb{E}\left[\beta_2(t) \frac{D_B^{(1)}(t) + 2 M^{(1)}_{B,**}(t)}{t}\right] \nonumber\\
	&  + \mathbb{E}\left[\beta_3(t)\left(1 +\frac{D_B^{(1)}(t) + 2 M^{(1)}_{B,**}(t)}{t}\right) + \beta_4(t)\right]\\
	&=\mathbb{E}[M^{(1)}_{B,**}(t)] + \frac{2\beta_2 + 2\beta_3}{t} \mathbb{E}[M^{(1)}_{B,**}(t)] \left(1 + o(1)\right)\\
	&= O\left(t^{2\beta_2 + 2\beta_3}\right).
\end{align}
The last step follows by Lemma \ref{lem-recursion-order}.

For (\ref{B,*B}), we can again write a recursive formula:
\begin{align}
	&\mathbb{E}[M^{(1)}_{B,*B}(t+1) \mid \mathcal{F}_t]
	= M^{(1)}_{B,*B}(t)  + \beta_2(t)\frac{M_{B, *B}^{(1)}(t)}{t} \\&+ \beta_3(t)\left(1 + \frac{D_B(t) + M_{B, **}^{(1)}(t) + M_{B, *B}^{(1)}(t)}{t}\right) 
\end{align}
and take the expection:
\begin{align}
	&\mathbb{E}[M^{(1)}_{B,*B}(t+1)]
	= \mathbb{E}\left[ M^{(1)}_{B,*B}(t)  + \beta_2(t)\frac{M_{B, *B}^{(1)}(t)}{t}\right. \\
	&+ \left.\beta_3(t)\left(1 + \frac{D_B(t) + M_{B, **}^{(1)}(t) + M_{B, *B}^{(1)}(t)}{t}\right)\right]\\
	& = \mathbb{E}[M^{(1)}_{B,*B}(t)]  + \left( \frac{\beta_2 + \beta_3}{t} \mathbb{E}[M^{(1)}_{B,*B}(t)] + \frac{\beta_3}{t} \mathbb{E}[M^{(1)}_{B,**}(t)]\right)(1+o(1))\\
	& = \frac{\beta_3}{\beta_2 + \beta_3} \mathbb{E} [M^{(1)}_{B,**}(t) (1+o_t(1)) ].
\end{align}
The last step again follows by Lemma \ref{lem-recursion-order}.

Finally, we repeat the same step as in (\ref{B,**}) for (\ref{R,**}), which gives
\begin{align}
	\mathbb{E}[M^{(1)}_{R,**}(t+1)]
	&=O\left(t^{2\beta_1 + 2\beta_4}\right).
\end{align}
By Lemma \ref{compare_beta}, we know $\mathbb{E} [M^{(k)}_{R,**}(t) ] = o_t(\mathbb{E} [M^{(k)}_{B,**}(t) ])$. We therefore complete the proof of the basis case.

Now assume that the claims hold for $\{2, \ldots, k-1 \}$. The idea is to write out the recursive equation for $M^{(k)}_{B,B*}(t)$, $M^{(k)}_{B,R*}(t)$,$M^{(k)}_{B,*B}(t)$, $M^{(k)}_{B,*R}(t)$, according to the 4 possible cases at each step. 

For $M^{(k)}_{B,B*}(t)$, notice that $M^{(k)}_{B,B*}(t+1)$ can arise from $M^{(k)}_{B,B*}(t)$ in the following cases:
\begin{enumerate}
	\item A new blue node connects with an existing red node $i$: then the new node contributes $1\cdot d_{i, B^*}^{(k-1)}(t)$ increments.
	\item A new blue node connects with an existing blue node $i$: the new node contributes $1\cdot d_{i, B^*}^{(k-1)}(t)$ increments, node $i$ contributes $1\cdot d_{i, B^*}^{(k)}(t)$ increments, and each blue node $j$ with $dist_{i,j}(t)= k-1$ contributes $d_j(t)$ increments, which sum up to $d_{i, B^*}^{(k)}(t)$ in total.
	\item A new red node connects with an existing blue node $i$: node $i$ contributes $1\cdot d_{i, B^*}^{(k)}(t)$ increments, and blue nodes that are $(k-1)$ away from $i$ contribute $d_j(t)$ increments in total.
	\end{enumerate}
Therefore,
\begin{align}
	\label{eq-recur-B-B*-1}
	&\mathbb{E}[M^{(k)}_{B,B*}(t+1) \mid \mathcal{F}_t]
	= M^{(k)}_{B,B*}(t) + 
	\beta_2(t) \frac{2 M^{(k)}_{B,B*}(t) + M^{(k-1)}_{B,*B}(t)}{t} \\
	& + \beta_3(t) \frac{2 M^{(k)}_{B,B*}(t) + M^{(k-1)}_{B,*B}(t) + M^{(k-1)}_{B,B*}(t)}{t}
	+ \beta_4(t) \frac{M^{(k-1)}_{R,B*}(t)}{t}.
\end{align}

Now, taking expectations on both sides of (\ref{eq-recur-B-B*-1}), we get
\begin{align}
	\label{eq-recur-B-B*-2}
	&\mathbb{E}[M^{(k)}_{B,B*}(t+1)]
	= \mathbb{E}[ M^{(k)}_{B,B*}(t)] + 
	\mathbb{E}\left[\beta_2(t) \frac{2 M^{(k)}_{B,B*}(t) + M^{(k-1)}_{B,*B}(t)}{t}\right] \nonumber \\
	& + \mathbb{E}\left[\beta_3(t) \frac{2 M^{(k)}_{B,B*}(t) + M^{(k-1)}_{B,*B}(t) + M^{(k-1)}_{B,B*}(t)}{t}\right]
	+ \mathbb{E}\left[\beta_4(t) \frac{M^{(k-1)}_{R,B*}(t)}{t}\right] \\
	&= \mathbb{E} [M^{(k)}_{B,B*}(t)] + (2 \beta_2 + 2 \beta_3) (1+o_t(t^{-1/4})) \frac{\mathbb{E} [M^{(k)}_{B,B*}(t)]}{t}\nonumber \\
		&+ 2 \beta_3 \mathbb{E} [M^{(k-1)}_{B,**}(t)] (1+o_t(1)),
\end{align}
where in the last step we used Lemma \ref{lem_exp_ratio} and the induction assumptions that
\begin{align}
		\mathbb{E} [M^{(k-1)}_{B,B*}(t)] &= \mathbb{E} [M^{(k-1)}_{B,**}] (1+o_t(1)),\\
		\mathbb{E} [M^{(k-1)}_{B,*B}(t)] &= \frac{\beta_3}{\beta_2 + \beta_3} \mathbb{E} [M^{(k-1)}_{B,**}] (1+o_t(1)),\\
		\mathbb{E} [M^{(k-1)}_{R,**}(t)] &= o_t(\mathbb{E} [M^{(k-1)}_{B,**}(t)]).
\end{align}

Note that by induction we also have
$$\mathbb{E} [M^{(k-1)}_{B,**}(t)] = O((\log t)^{k-1} t^{2 \beta_2 + 2 \beta_3}).$$
Applying Lemma \ref{lem-recursion-order} on $\mathbb{E} [M^{(k)}_{B,B*}(t)]$, we are in the cast that $c_1 = c_3 = 2 \beta_2 + 2 \beta_3$ and $m = k-1$, and thus we see that
\begin{equation}
	\mathbb{E} [M^{(k)}_{B,B*}(t)]= O((\log t)^{k} t^{2 \beta_2 + 2 \beta_3}).
\end{equation}

We can then apply the same method on $M^{(k)}_{B,R*}(t)$,$M^{(k)}_{B,*B}(t)$, $M^{(k)}_{B,*R}(t)$, and get the desired result.
\end{proof}

\begin{lemma}\label{compare_beta}
Let $\beta_1(t), \beta_2(t),\beta_3(t), \beta_4(t)$ defined as in the previous proof, and $\beta_i:= \lim_{t\rightarrow \infty} \beta_i(t)$. Then
\begin{equation}
\beta_1 + \beta_4< \beta_2 + \beta_3.
\end{equation}

\begin{proof}
We first simplify the expression by defining $g(x,y):= \frac{x}{1-(1-x)(1-y)}$, then $\beta_1 = g(\alpha^*,\rho)$, $\beta_2 = 1- g(\alpha^*,\rho)$, $\beta_3 = g(1-\alpha^*,\rho)$, and $\beta_r = 1-g(1-\alpha^*,\rho)$. Furthermore, we define $h(x,y):= \frac{g(x,y)}{x}$. Then $\frac{1-(g(x,y)}{1-x} = y\cdot h(x,y)$. Now,
\begin{equation}
\lim_{t\rightarrow \infty} \frac{\beta_1(t) - \beta_2(t) }{\beta_3(t)  - \beta_4(t) } = \frac{\frac{r}{2}(1-\rho)h\left(\alpha^*,\rho\right)}{\frac{1-r}{2}(1-\rho)h\left(1-\alpha^*,\rho\right)} = \frac{rh\left( \alpha^*,\rho\right)}{(1-r)h\left(1- \alpha^*,\rho\right)}.
\end{equation}
It suffices to show that the above quantity is smaller than 1. Theorem 4.6 in \cite{avin2015homophily} has shown that as $t\rightarrow \infty$,
\begin{align*}
2 \alpha^* - 1 &= rg\left( \alpha^* , \rho\right)- (1-r)g\left(1-\alpha^*, \rho\right)\\ \alpha^* - (1-\alpha^*) &= rh\left( \alpha^*,\rho\right) \alpha^* - (1-r)h(1- \alpha^*,\rho)\left(1- \alpha^*\right)\\
\alpha^*\left(1- rh(\alpha^*,\rho)\right)& = (1- \alpha^*)\left(1-(1-r)h\left(1- \alpha^*,\rho\right)\right)
\end{align*}
By Part 5 of Theorem 4.4, $\alpha^* < 1/2$. Therefore,
\begin{equation}
1- rh\left( \alpha^*,\rho\right)>1-(1-r)h\left(1- \alpha^*,\rho\right).
\end{equation}
Thus, $rh\left( \alpha^*,\rho\right) < (1-r)h\left(1- \alpha^*,\rho\right)$ and $\beta_1-\beta_2 < \beta_3 - \beta_4$ as $t\rightarrow \infty$.
\end{proof}
\end{lemma}

\subsection{Proof of Lemma \ref{lem_exp_ratio}}

\begin{replemma}{lem_exp_ratio}
	For any color $C \in \{R, B\}$, any $i \in \{1,2,3,4\}$ and any $k$, we have that
	\begin{equation}
		\frac{\mathbb{E}\left[ \beta_i(t) M^{(k-1)}_{C,**}(t) \right]}
		{\beta_i \mathbb{E}\left[ M^{(k-1)}_{C,**}(t) \right]}
		= 1 + o_t(t^{-1/4}).
	\end{equation}
\end{replemma}
\begin{proof}
	From \cite{avin2015homophily}[Corollary 4.11], denoting by 
	\begin{equation}\sigma_t = \max(2 (\log t) t^{-1/2}, t^{-1/3}),
	\end{equation}	
 we have that
	\begin{equation}
		\mathbb{P}\left( |\alpha(t) - \alpha^*| > \sigma_t \right)
		< t^{-4}.
	\end{equation}
	Based on the above result, we can check that there exists a constant $C_1 > 0$, such that for any $i \in \{1,2,3,4\}$,
	\begin{equation}
		\mathbb{P}\left( |\beta_i(t) - \beta_i| > C_1 \sigma_t \right)
		< t^{-4}.
	\end{equation}
	Note that by definition, $M^{(k-1)}_{C,**}(t)$ is upper bounded by $t^3$; and there exists a constant $C_2 > 0$ such that $\beta_i(t), \beta_i < C_2$ for any $i,t$. We have
	\begin{align}
	&	\left \vert \mathbb{E}\left[ (\beta_i(t)-\beta_i) M^{(k-1)}_{C,**}(t) \right]
		\right \vert
	\leq 
		\mathbb{E}\left[ \vert \beta_i(t)-\beta_i \vert M^{(k-1)}_{C,**}(t) \right] \\
	&\leq \mathbb{E}\left[ \vert \beta_i(t)-\beta_i \vert M^{(k-1)}_{C,**}(t) \cdot 1_{|\beta_i(t) - \beta_i| > C_1 \sigma_t} \right] \nonumber \\
		&+
		\mathbb{E}\left[ \vert \beta_i(t)-\beta_i \vert M^{(k-1)}_{C,**}(t)\cdot 1_{|\beta_i(t) - \beta_i| \leq C_1 \sigma_t} \right] \\
	& \leq C_2 t^3 \cdot t^{-4} + C_1 \sigma_t \mathbb{E}\left[ M^{(k-1)}_{C,**}(t) \right].
	\end{align}
	The above bound implies that for some $C_2 > 0$
	\begin{equation}
		\frac{\left \vert \mathbb{E}\left[ (\beta_i(t)-\beta_i) M^{(k-1)}_{C,**}(t) \right]
		\right \vert}
		{\beta_i \mathbb{E}\left[ M^{(k-1)}_{C,**}(t) \right]}
		\leq C_2 \sigma_t = o_t(t^{-1/4}),
	\end{equation}
	where $-1/4$ is just an arbitrary number between $-1/3$ and $0$. The proof is finished.
\end{proof}

\subsection{Proof of Lemma \ref{lem-recursion-order}}

\begin{replemma}{lem-recursion-order}
	For a sequence $\{a_t\}$ satisfying
	\begin{equation}
		a_{t+1} = a_t + c_1 d_t \frac{a_t}{t} + c_2 b_t, 
	\end{equation}
	where $d_t = (1+o_t(t^{-\epsilon}))$ for some $\epsilon > 0$, $b_t = (\log t)^m t^{c_3 - 1} (1+o_t(1))$. 
	We then have that
	\begin{equation}
		a_t = 
		\begin{cases}
			O(t^{c_1}), &\text{if $c_1 > c_3$} \\
			O((\log t)^{m+1} t^{c_3}), &\text{if $c_1 = c_3$} \\
			\frac{c_2}{c_3-c_1} (\log t)^m t^{c_3} (1+o_t(1)), &if c_1 < c_3.
		\end{cases}
	\end{equation}
\end{replemma}

\begin{proof}
	By the recursive equation for $a_t$, we see that
	\begin{align}
		a_t = \sum_{s=1}^{t-1} c_2 b_s \prod_{j=s+1}^{t-1} (1+c_1d_j/j).
	\end{align}
	We can find a constant $C_0 > 0$, such that as $c_1d_j/j < 1/2$, 
	$$\vert \log(1+c_1d_j/j) - c_1d_j/j \vert < C_0 (c_1d_j/j)^2,$$
	and thus
	\begin{align}
		& \left \vert \sum_{j=s+1}^{t-1} \log  (1+c_1d_j/j)
		-  \sum_{j=s+1}^{t-1} c_1/j \right \vert \\
		\leq & \left \vert \sum_{j=s+1}^{t-1} c_1 (1-d_j)/j \right \vert
		+ \left \vert \sum_{j=s+1}^{t-1} C_0 (c_1d_j/j)^2 \right \vert.
	\end{align}
	By the assumption that $d_t = (1+o_t(t^{-\epsilon}))$, it is easy to check that 
	\begin{align}
		&\left \vert \sum_{j=s+1}^{t-1} c_1 (1-d_j)/j \right \vert
		+ \left \vert \sum_{j=s+1}^{t-1} C_0 (c_1d_j/j)^2 \right \vert\\
		\leq &
		\left \vert \sum_{j=s+1}^{\infty} c_1 o_j(j^{-1-\epsilon}) \right \vert
		+ \left \vert \sum_{j=s+1}^{\infty} C_0 (c_1/j)^2 (1+o_j(1)) \right \vert
		 = o_s(1).
	\end{align}
	Also, by basic mathematical analysis we have that
	\begin{equation}
		\left \vert \sum_{j=s+1}^{t-1} c_1/j - c_1(\log(t) - \log(s))\right \vert
		= o_s(1).
	\end{equation}
	Hence we see that
	
	\begin{align}
		& \prod_{j=s+1}^{t-1} (1+c_1d_j/j) 
		= \exp\left( \sum_{j=s+1}^{t-1} \log  (1+c_1d_j/j) \right) \\
		& = \exp\left( \sum_{j=s+1}^{t-1} c_1/j + o_s(1) \right)
		= (t/s)^{c_1}(1+o_s(1)).
	\end{align}
	
	Combining with the expression of $b_t$, we get
	\begin{align}
		a_t & = \sum_{s=1}^{t-1} c_2 (\log s)^m s^{c_3 - 1} (t/s)^{c_1} (1+o_s(1)) \\
		&	= t^{c_1} \sum_{s=1}^{t-1} c_2 (\log s)^m s^{c_3 - c_1 - 1} (1+o_s(1)).
	\end{align}
	We have that
	\begin{align}
	&	\sum_{s=1}^{t-1} c_2 (\log s)^m s^{c_3 - c_1 - 1} (1+o_s(1)) \\
		= &
		\begin{cases}
			O(1), &\text{if $c_1 > c_3$} \\
			O((\log t)^{m+1}), &\text{if $c_1 = c_3$} \\
			\frac{c_2}{c_3 - c_1} (\log t)^m t^{c_3 - c_1}(1+o_t(1)), &if c_1 < c_3,
		\end{cases}
	\end{align}
	which finishes the proof.
\end{proof}

\end{document}